\documentclass[11pt,envcountsect,orivec]{llncs}

\usepackage{graphicx,color}
\usepackage{latexsym,amsfonts}
\usepackage{amsmath}
\usepackage{amssymb}
\usepackage{bbm}
\usepackage{array}

\usepackage{algorithm}
\usepackage{algorithmicx}
\usepackage[noend]{algpseudocode}
\usepackage{tikz}
\usepackage{url}
\usepackage{fullpage}
\usepackage{enumitem}

\usepackage{thmtools}
\usepackage{thm-restate}

\definecolor{blue}{RGB}{0,82,147} 
\definecolor{red}{RGB}{202,033,063}
\definecolor{green}{RGB}{98,158,31}

\usepackage[breaklinks=true]{hyperref} 
\hypersetup{
	colorlinks=true,
	citecolor=green,
	linkcolor=blue,  
	urlcolor=black,
}
\newcounter{lpnumber} 

\newcounter{dummy}
\makeatletter
\newcommand\myitem[1][]{\item[#1]\refstepcounter{dummy}\def\@currentlabel{#1}}
\makeatother

\newcommand{\myproblem}[1]{{\sf{#1}}}
\newcommand{\wt}{\mathsf{wt}}
\newcommand{\Nbr}{\mathsf{Nbr}}
\newcommand{\vote}{\mathsf{vote}}
\renewcommand{\k}{t}

\title{\Large \texorpdfstring{The popular assignment problem: \\ when cardinality is more important than popularity}{The popular assignment problem: when cardinality is more important than popularity}}
\author{Telikepalli Kavitha\inst{1} \and Tam\'{a}s Kir\'{a}ly\inst{2} \and Jannik Matuschke\inst{3} \and \\ Ildik\'{o} Schlotter\inst{4,5} \and Ulrike Schmidt-Kraepelin\inst{6}}

\institute{Tata Institute of Fundamental Research, Mumbai, India; \email{kavitha@tifr.res.in} 
\and E\"otv\"os Lor\'and University, Budapest, Hungary; \email{tamas.kiraly@ttk.elte.hu} 
\and Research Center for Operations Management, KU Leuven, Leuven, Belgium; \email{jannik.matuschke@kuleuven.be} 
\and Centre for Economic and Regional Studies, Budapest, Hungary; \email{schlotter.ildiko@krtk.hun-ren.hu} 
\and  Budapest University of Technology and Economics, Budapest, Hungary; 
\and Simons Laufer Mathematical Sciences Institute (SLMath), Berkeley, California, United States; \email{uschmidt@slmath.org}
}

\begin{document}

\date{}
\pagestyle{plain}
\maketitle
\begin{abstract}
We consider a matching problem in a bipartite graph $G=(A\cup B,E)$ where nodes in $A$ are
agents having preferences in partial order over their neighbors, while nodes in $B$ are objects without preferences. 
We propose a polynomial-time combinatorial algorithm based on LP duality that finds a maximum matching or \emph{assignment} in $G$  that is popular among all maximum matchings, if there exists one. Our algorithm can also be used to achieve a trade-off between popularity and cardinality by imposing a penalty on unmatched nodes in $A$.

We also provide an $O^*(|E|^k)$ algorithm that finds an assignment whose unpopularity margin is at most~$k$; this algorithm is essentially optimal, since the problem is $\mathsf{NP}$-complete and $\mathsf{W}_l[1]$-hard with parameter~$k$. 
We also prove that finding a popular assignment of minimum cost when each edge has an associated binary cost is $\mathsf{NP}$-hard, even if agents have strict preferences. 
By contrast, we propose a polynomial-time algorithm for the variant of the popular assignment problem
with forced/forbidden edges. Finally, we present an application in the context of housing markets.
\end{abstract}

\section{Introduction}
\label{sec:intro}
We consider a matching problem in a bipartite graph $G = (A \cup B,E)$ with {\em one-sided} preferences. Nodes in~$A$, also
called {\em agents}, have preferences in partial order over their neighbors, while nodes in~$B$, also called {\em objects}, have no preferences.
The preference order $\succ_a$ of an agent~$a \in A$ over its neighboring objects in~$G$ is a \emph{partial order} if $\succ_a$ is irreflexive and transitive.
The model of one-sided preferences is often called the {\em house allocation} problem as it arises in campus housing allocation in universities~\cite{AS98}. The fact that
preferences are {\em one-sided} here makes this model very different from the {\em marriage} problem introduced by Gale and Shapley~\cite{GS62} in 1962, where
{\em all} nodes have preferences over their neighbors.

Usually one seeks a matching in $G$ that is {\em optimal} in some sense. 
Popularity is a well-studied notion of optimality in the model of one-sided preferences. Any pair of matchings, 
say $M$ and $N$, can be compared by holding an election between them where agents are voters. Every agent prefers the matching where she gets assigned 
a more preferred partner, and being unmatched is her worst choice. 
Let $\phi(M,N)$ be the number of agents who prefer $M$ to~$N$. Then we 
say that $M$ is \emph{more popular than}~$N$ if $\phi(M,N) > \phi(N,M)$.
Let us write $\Delta(M,N) = \phi(M,N) - \phi(N,M)$. 

\begin{definition}
  \label{def:popular-matching}
  A matching $M$ is {\em popular} if there is no matching more popular than $M$, i.e., $\Delta(M,N) \ge 0$ for all matchings $N$ in $G$.
\end{definition}  

The \myproblem{popular matching} problem involves deciding if $G$ admits a popular matching, and finding one if so. 
This is a well-studied problem from 2005, 
and  there is an efficient algorithm to solve it
in the case when each agent's preference order is a 
\emph{weak ranking}~\cite{AIKM07}, i.e., when the indifference\footnote{We say that agent~$a \in A$ is \emph{indifferent} between objects~$b$ and~$b'$ if $b \not \succ_a b'$ and $b' \not \succ_a b$.} relation is transitive.
Our goal in this paper is to solve the following natural generalization of the \myproblem{popular matching} problem when agents have partial order preferences (these generalize weak rankings).

\paragraph{\bf The popular assignment problem.}
Consider applications where the {\em size} of the matching is of primary importance. It is natural that as many students as possible be assigned campus housing. Another application is in assigning final year medical and nursing students to hospitals during emergencies (such as a pandemic) to overcome staff shortage~\cite{ET21}. Preferences of these students are important but the size of the matching is more important, since we want to augment human resources as much as possible. Thus what we seek is not a popular matching but a popular {\em maximum matching}, i.e., among maximum matchings, a best one. 
Our approach to prioritize the cardinality of the matching is in stark contrast with most existing results in the area of popular matchings, 
where the foremost requirement is usually popularity.

By augmenting $G$ with dummy agents and artificial objects (see Section~\ref{sec:prelims}), we can assume that~$G$ admits a perfect matching, i.e., an \emph{assignment}. So
our problem becomes the popular perfect matching problem---we will call this the \myproblem{popular assignment} problem in $G$.
In other words, we seek  an assignment of objects to agents such that every agent is assigned an object and, roughly speaking, there is no assignment 
that makes more agents happier than it makes unhappier.

\begin{definition}
\label{def:assignment}
A perfect matching $M$ is a {\em popular assignment} if there is no perfect matching in~$G$ that is more popular than $M$, i.e.,
$\Delta(M,N) \ge 0$ for all perfect matchings~$N$ in~$G$.
\end{definition}

Thus, a popular assignment is a {\em weak Condorcet winner}~\cite{Con85,condorcet} where all perfect matchings are candidates and agents are voters.
Weak Condorcet winners need not exist in a general voting instance; in our setting as well, a popular assignment need not exist in $G$.
Consider the following simple example where $A = \{a_1,a_2,a_3\}$ and $B = \{b_1,b_2,b_3\}$, and $G$
is the complete bipartite graph~$K_{3,3}$, i.e., every agent and object are adjacent. Suppose that every agent has the same (strict) preference ordering:
$b_1 \succ b_2 \succ b_3$, i.e., $b_i$ is the $i$-th choice for $i = 1,2,3$. 
For any assignment here, there exists a more popular assignment, e.g., $\{(a_1,b_3),(a_2,b_1),(a_3,b_2)\}$ is more popular than 
$\{(a_1,b_1),(a_2,b_2),(a_3,b_3)\}$. Therefore, this instance has no popular assignment.

\paragraph{\bf The \myproblem{popular assignment} problem.}
Given a bipartite graph $G = (A \cup B, E)$ where every $a \in A$ has preferences in partial order
over her neighbors, does $G$ admit a popular assignment? If so, find one.

\smallskip

It is easy to show instances  that admit popular assignments but do not have any popular matching (see Section~\ref{app:no-pop-assignment}). Interestingly, an algorithm for the \myproblem{popular assignment} problem also solves the \myproblem{popular matching} problem. By augmenting the given instance with artificial {\em worst-choice} objects and some dummy agents, we can construct an instance $G'$ on at most twice as many nodes as in $G$ such that~$G$ admits a popular matching if and only if $G'$ admits a popular assignment (this simple reduction is given in Section~\ref{app:pop-matching-reduction}). Thus, the \myproblem{popular assignment} problem generalizes the \myproblem{popular matching} problem. 
Furthermore, we introduce a family of problems that interpolates between these two problems by introducing a penalty (parameterized by~$\k$) for unmatched agents. In particular, the popular matching problem with penalty $\k=1$ equals the \myproblem{popular matching} problem, while the same problem with $\k=n$ corresponds to the \myproblem{popular assignment} problem. In fact, we can reduce all of these problems to the \myproblem{popular assignment} problem (see Section~\ref{sec:k-level-algo}).

By adjusting the usage of worst-choice objects appropriately, an algorithm for \myproblem{popular assignment} can solve 
the following more general variant of  both the \myproblem{popular matching} problem and the \myproblem{popular assignment} problem, and thus opens possibilities to a wide spectrum of applications.

\paragraph{\bf Popularity with diversity.} 
Consider instances~$G = (A \cup B, E)$ where every agent has one of $k$ colors associated with it, and we are interested in only those (not necessarily perfect) matchings that match for every $i \in \{1,\ldots,k\}$, 
$c_i$ agents of color~$i$, where $s_i \le c_i \le t_i$ for some given bounds~$s_i$ and~$t_i$, i.e., only those matchings that satisfy these lower and upper bounds for every color are admissible. We seek a matching that is popular {\em within} the set of admissible matchings (see Section~\ref{app:pop-matching-reduction} for a reduction to \myproblem{popular assignment}).

Public housing programs constitute an application where such problems arise.
For example, in Singapore, $70\%$ of the residential real estate is managed by a public housing program which promotes ethnic diversity by imposing quotas on each housing block and ethnic group. Motivated by this market, Benabbou et al. \cite{BCH+18a} study a similar model with cardinal utilities.

\subsection{Our contribution}
Our first result is that the \myproblem{popular assignment} problem (with general partial order preferences) can be solved in polynomial time. Let $|A| = |B| = n$ and $|E| = m$.

\begin{theorem}
  \label{thm:algo}
  The \myproblem{popular assignment} problem in $G = (A \cup B, E)$ can be solved in $O(m\cdot n^{5/2})$ time. 
\end{theorem}

Note that Theorem~\ref{thm:algo} implies that the \myproblem{popular matching} problem with general partial order preferences can be solved in polynomial time, using the reduction described in Section~\ref{app:pop-matching-reduction}. 
Curiously, the popular matching algorithm in \cite{AIKM07} does not generalize to partial order preferences; we include an example that illustrates this in Section~\ref{app:characterization-of-pop-matchings}.
To the best of our knowledge, no polynomial-time algorithm for solving the \myproblem{popular matching} problem with partial order preferences has been known prior to Theorem~\ref{thm:algo}.

Partial order preferences arise naturally, e.g., when agents have cardinal utilities for the objects that then induce \emph{semiorder} preferences.
This means that
for any agent, a pair of objects with widely differing utilities get compared by their utilities, while a pair of objects with utilities that are within a small threshold 
are deemed equivalent, thus the agent is indifferent between two such objects.\footnote{Yu Yokoi posed the \myproblem{popular matching} problem with semiorder preferences as an open problem in a talk on ``Approximation Algorithms for Matroidal and Cardinal Generalizations of Stable Matching'' at RIMS, Kyoto in May 2023 (she was not aware that the conference version of our paper~\cite{KKMSS-soda2022} solves the \myproblem{popular matching} problem with partial order preferences).}  Observe that the indifference relation here need not be transitive.

When a popular assignment does not exist in~$G$, a natural extension is to ask for an {\em almost} popular assignment, i.e., an assignment with {\em low}
unpopularity. How do we measure the unpopularity of an assignment? A well-known measure is the {\em unpopularity margin}~\cite{McC08} 
defined for any assignment~$M$ as $\mu(M) = \max_N (\phi(N,M) - \phi(M,N))=\max_N \Delta (N,M)$,
where the maximum is taken over all assignments, that is, all perfect matchings~$N$ in~$G$. Thus $\mu(M)$ is the maximum margin by which another assignment defeats~$M$. 

An assignment~$M$ is popular if and only if $\mu(M) = 0$. Let the \myproblem{$k$-unpopularity margin} problem be the problem of deciding if $G$ admits an 
assignment with unpopularity margin at most~$k$. We generalize Theorem~\ref{thm:algo} to show the following result.

\begin{theorem}
\label{thm:unpopmargin-algo-xp}
For any $k \in \mathbb{Z}_{\ge 0}$, the \myproblem{$k$-unpopularity margin} problem in $G=(A \cup B,E)$ can be solved in $O(m^{k+1} \cdot n^{5/2})$ time.
\end{theorem}

So when $k = \Theta(1)$, this is a polynomial-time algorithm to find a $k$-unpopularity margin assignment 
if $G$ admits one.
Via the reduction described in Section~\ref{app:pop-matching-reduction}, which preserves the unpopularity margin of any matching, Theorem~\ref{thm:unpopmargin-algo-xp} extends to the \myproblem{$k$-unpopularity margin} version of the \myproblem{popular matching} problem, i.e., to decide if $G$ admits a matching with unpopularity margin at most~$k$.
Though the measure of unpopularity margin was introduced in 2008, no polynomial-time algorithm was known so far to solve the $k$-unpopularity margin matching
problem for any $k \ge 1$.

Rather than the exponential dependency on the parameter~$k$ in Theorem~\ref{thm:unpopmargin-algo-xp}, 
can we solve the \myproblem{$k$-unpopularity margin} problem in polynomial time? 
Or at least can we achieve a running time of the form~$f(k)\mathsf{poly}(m,n)$ for some function~$f$ so that the degree of the polynomial is independent of~$k$?
That is, can we get a \emph{fixed-parameter tractable} algorithm with parameter~$k$?
The following hardness result shows that the algorithm of Theorem~\ref{thm:unpopmargin-algo-xp} is 
essentially optimal for the \myproblem{$k$-unpopularity margin} problem. See Section~\ref{hardness:k-unpop-margin} for the definition of
$\mathsf{W}_l[1]$-hardness.

\begin{theorem}
\label{thm:unpopmargin-hardness}
The \myproblem{$k$-unpopularity margin}  problem is $\mathsf{W}_l[1]$-hard with parameter~$k$ when agents' preferences are weak rankings, 
and it is $\mathsf{NP}$-complete 
even if preferences are strict rankings.
\end{theorem}

\paragraph{\bf Popularity with penalty for unmatched agents.}
We also 
consider a model that allows for a trade-off between popularity and cardinality by imposing a \emph{penalty} for unmatched agents. A matching $M$ in $G = (A \cup B, E)$ is \emph{popular with penalty $\k$} if it does not lose a pairwise 
election against
any matching $N$ in $G$, where votes of agents that are matched in $M$ but unmatched in $N$ (or vice versa) are counted with weight~$\k$; 
see Section~\ref{sec:k-level-algo} for a formal definition. 
If $M$ is a matching that is popular with penalty $\k$, then $|M| \ge \frac{\k}{\k+1}|M_{\max}|$ 
where $M_{\max}$ is a maximum matching in $G$.

We prove that a {\em matching} that is popular with penalty $\k$, if it exists, can be found via an appropriately constructed instance of the \myproblem{popular assignment} problem (without penalty). Moreover, we show that a truncated variant of our algorithm for the \myproblem{popular assignment} problem can find an {\em assignment} that is popular with penalty~$\k$ (when compared against any matching), if such an assignment exists.

\paragraph{\bf Minimum-cost popular assignment.} We next consider the \myproblem{minimum-cost popular assignment} problem: 
suppose that, in addition to our instance $G=(A \cup B,E)$ of the \myproblem{popular assignment} problem 
where every agent has preferences in partial order
over her adjacent objects,
there is a cost function~$c:E \rightarrow \mathbb{R}$ on the edges and a budget $\beta \in \mathbb{R}$, and we want to decide whether
$G$ admits a popular assignment whose sum of edge costs is at most $\beta$.
Computing a minimum-cost popular assignment efficiently would also imply an efficient algorithm for finding a popular assignment with \emph{forced/forbidden} edges. We show the following hardness result.

\begin{theorem}
\label{thm:mincost-nphard-01}
The \myproblem{minimum-cost popular assignment} problem is $\mathsf{NP}$-complete, even if each edge cost is in $\{0,1\}$
and agents have strict preferences.
\end{theorem}

This result reveals a surprising difference between the \myproblem{popular assignment} problem and the \myproblem{popular matching} problem; 
for the latter problem, arbitrary linear cost functions can be minimized efficiently when preferences are
weak rankings~\cite{AIKM07}.
The $\mathsf{NP}$-hardness for the \myproblem{minimum-cost popular assignment} problem implies that given a set of desired edges,
it is $\mathsf{NP}$-hard (even for strict preferences) to find a popular assignment that contains the maximum number of desired edges.
Interestingly, in spite of the above hardness result, 
we show the following positive result that
the \myproblem{popular assignment} problem with partial order preferences and forced/forbidden edges is tractable. 
Note that the assignment $M$ must be popular among \emph{all} assignments, not only those adhering to the forced and forbidden edge constraints.
Thus, although the {\em desired} edges variant is hard, the {\em forced} edges variant is easy.

\begin{theorem}
\label{thm:forced-forbidden}
Given a set $F^+ \subseteq  E$ of forced edges and another set $F^- \subseteq  E$ of forbidden edges, we can determine in polynomial
time if there exists a popular assignment $M$ in $G = (A \cup B, E)$ such that $F^+ \subseteq M$ and $F^- \cap M = \emptyset$.
\end{theorem}

Thus the \myproblem{popular assignment} problem is reminiscent of the well-known \myproblem{stable roommates} problem\footnote{This problem asks for a stable matching 
in a general graph (which need not be bipartite) with strict preferences.};
in a roommates instance,  
finding a stable matching  can be solved in polynomial time~\cite{Irv85} even with forced/forbidden edges~\cite{FIM07}, however
finding a minimum-cost stable matching 
is $\mathsf{NP}$-hard~\cite{Fed92}.

\paragraph{\bf Application to housing markets.}
Finally, in Corollary~\ref{prop:pop_allocation} we present an application of the algorithm of Theorem~\ref{thm:algo} 
that allows us to find popular allocations in so-called \emph{housing markets}, a model defined by Shapley and Scarf~\cite{SS74};
see Section~\ref{sec:housing} for all definitions concerning housing markets. 
We contrast this with Theorem~\ref{thm:housealloc-nphard}, a strengthening of Theorem~\ref{thm:mincost-nphard-01},
showing that it is $\mathsf{NP}$-hard to find a popular allocation that maximizes the number of trading agents in the housing market.

\paragraph{\bf Conference version.}
A preliminary version of this paper appeared in the proceedings of SODA 2022~\cite{KKMSS-soda2022}.
This version of the paper goes beyond the conference version in the following two points: (i) Section~\ref{sec:2-level-algo}, Section~\ref{sec:k-level-algo}, and Section~\ref{sec:housing} are new and (ii) Theorem~\ref{thm:mincost-nphard-01} is proved in its full generality; in fact, we prove a stronger version of Theorem~\ref{thm:mincost-nphard-01}, namely Theorem~\ref{thm:housealloc-nphard}. The conference version~\cite{KKMSS-soda2022} contained the proof of a weaker result which stated that the \myproblem{minimum-cost popular assignment} problem is $\mathsf{NP}$-complete when edge costs are in $\{0,1,+\infty\}$. 

\subsection{Background}
The notion of popularity in a marriage instance (where preferences are {\em two-sided} and strict) was introduced by G{\"a}rdenfors~\cite{Gar75} in 1975. 
Popular matchings always exist in such an instance, since any stable matching is popular~\cite{Gar75}. When preferences are one-sided, popular matchings 
need not exist. 
For the case when preferences are weak rankings, a simple and  clean combinatorial characterization of popular matchings (see Section~\ref{app:characterization-of-pop-matchings}) was given in \cite{AIKM07}, 
leading to an $O(m\sqrt{n})$ time algorithm~\cite{AIKM07} for the \myproblem{popular matching} problem.  
By contrast, a combinatorial characterization of popular assignments (even for strict rankings) remains elusive.
Finding a minimum unpopularity margin matching was proved to be $\mathsf{NP}$-hard~\cite{McC08}.

In the last fifteen years, popularity has been a widely studied concept. 
Researchers have considered extensions of the \myproblem{popular matching} problem  
where one aims for a popular matching satisfying some additional optimality criteria such as rank-maximality or fairness~\cite{KN09,MI11}, 
or where the notion of popularity is adjusted to incorporate capacitated objects or weighted agents~\cite{Mestre14,SM10}.
Another variant of the popular matching problem was considered in \cite{CHK17} where nodes in $A$ have strict preferences and nodes in~$B$, i.e., 
objects, have no preferences, however each object cares to be matched to any of its neighbors. We refer to \cite{Cseh17} for a survey on results in this area. 

Among the literature on popular matchings, only a few studies have considered a setting that focuses on popularity
within a restricted set of admissible solutions.
The paper that comes closest to our work is \cite{Kav14} which considered the popular maximum matching problem
in a marriage instance (where preferences are two-sided and strict). It was shown there that a popular maximum matching always exists in a marriage instance and
one such matching can be computed in $O(mn)$ time; the following size-popularity trade-off in a marriage instance was also shown in \cite{Kav14}: for any integer $t\ge 2$, there exists a matching $M_t$ with $|M_t| \ge \frac{t}{t+1}|M_{\max}|$ and {\em unpopularity factor} at most $t-1$.\footnote{The unpopularity factor of a matching $M$ is $\max_N\phi(N,M)/\phi(M,N)$, where the maximum is taken over all matchings $N \ne M$ in $G$~\cite{McC08}.} 
Very recently, it was shown in \cite{Kav20} that a minimum-cost popular maximum matching in a marriage
instance can be computed in polynomial time. These results use the rich machinery of stable matchings in a marriage instance~\cite{GS62,Rot92}.
In contrast to these positive results for popular maximum matchings, computing an {\em almost-}stable maximum matching (one with the least number of blocking edges) 
in a marriage instance is $\mathsf{NP}$-hard~\cite{BMM10}. 

\subsection{Techniques} 
Our popular assignment algorithm is based on LP duality. We show that a matching~$M$ is a popular assignment if and only if it 
has a \emph{dual certificate} 
$\vec{\alpha} \in \{0, \pm 1,\ldots, \pm (n-1)\}^{2n}$ fulfilling certain constraints induced by the matching $M$. 
Our algorithm (see Section~\ref{sec:algo}) can be viewed as a search for such a dual certificate.
It associates a {\em level} $\ell(b)$ with every $b \in B$. 
This level function $\ell$ guides us in constructing a subgraph $G_{\ell}$ of $G$. 
If $G_{\ell}$ contains a perfect matching, then this matching is a popular assignment in $G$ and the levels determine a corresponding dual certificate.
If $G_{\ell}$ has no perfect matching, then we increase some $\ell$-values and update $G_{\ell}$ accordingly, until eventually $G_{\ell}$ contains a perfect matching or the level of an object increases beyond the permitted range, in which case we prove that no popular assignment exists.

The LP method for popular matchings was introduced in \cite{KMN09} and dual certificates for popular (maximum) matchings in marriage instances were shown in \cite{Kav16,Kav20}.
However, dual certificates for popular matchings in instances with one-sided preferences have not been investigated so far. The existence of simple dual certificates 
for popular assignments is easy to show (see Section~\ref{sec:prelims}), but this does not automatically imply polynomial-time solvability. The main novelty of our paper lies in showing a combinatorial algorithm to search for dual certificates in an instance $G$ and in 
using this approach to solve the \myproblem{popular assignment} problem, as well as several related problems, in polynomial time.

\paragraph{\bf Our other results.}
Our algorithm for the variant of the \myproblem{popular assignment} problem with {\em forced/ forbidden} edges (see Section~\ref{sec:edge-restrictions}) 
is a natural extension of the above algorithm where
certain edges are excluded. The \myproblem{$k$-unpopularity margin} algorithm (see Section~\ref{sec:min-unpopular}) associates a {\em load} with every edge such that the total load is at most $k$ and the overloaded edges are treated as {\em forced} edges. 
Our $W_l[1]$-hardness result shows that this $O^*(m^k)$ algorithm for the \myproblem{$k$-unpopularity margin} problem is essentially optimal, i.e., it is highly unlikely that this problem admits an $f(k)m^{o(k)}$ algorithm for any computable function~$f$. 

\section{Preliminaries}
For any $v \in A\cup B$, let $\Nbr_G(v)$ denote the set of neighbors of $v$ in $G$, and~$\delta(v)$ the set of edges incident to $v$.
For any $X \subseteq A\cup B$, we let $\Nbr_G(X)= \cup_{v \in X} \Nbr_G(v)$; 
we may omit the subscript $G$ if it is clear from the context. 
For any set $X$ of vertices (or edges) in~$G$, let~$G-X$ be the subgraph of~$G$ 
obtained by deleting the vertices (or edges, respectively) of $X$ from~$G$.
For an edge set $F \subseteq E$, we denote by $G[F]$ the subgraph of~$G$ restricted to edges of~$F$, that is, $G[F]=(A \cup B,F)$.
For a matching~$M$ in~$G$ and a node~$v$ matched in~$M$, we denote the partner of~$v$ by~$M(v)$.

The preferences of node $a \in A$
over its neighbors are given by a strict partial order~$\succ_a$, so $b \succ_a b'$ means that $a$ prefers $b$ to $b'$. We
use $b \sim_a b'$ to denote that $a$ is indifferent between $b$ and $b'$, i.e., neither $b \succ_a b'$ nor $b' \succ_a b$ holds. 
As mentioned earlier, the relation $\succ_a$ is a \emph{weak ranking} if $\sim_a$ is transitive. In this case, $\sim_a$ is an equivalence relation and there is a strict order on the equivalence classes. When each equivalence class has size~1, we call it a \emph{strict ranking} or a \emph{strict preference}.

We will also use the notation 
$[k]=\{1,2,\dots, k\}$ 
 for any positive integer~$k$.

\subsection{A characterization of popular matchings from \texorpdfstring{\cite{AIKM07}}{Abraham et al. (2007)} for weak rankings}\label{app:characterization-of-pop-matchings}
In order to characterize popular matchings 
in the case when agents' preferences are weak rankings, 
as done in \cite{AIKM07}, it will be convenient to add artificial worst-choice or last
resort objects to the given instance~$G = (A \cup B, E)$. So $B = B \cup \{l(a): a \in A\}$, i.e., corresponding to each~$a \in A$, a node~$l(a)$ 
gets added to~$B$ and we set this node $l(a)$ as the worst-choice object for~$a$. Thus we have $E = E \cup \{(a,l(a)): a \in A\}$.

We say that an edge $(a,b)$ \emph{dominates} an edge $(a,b')$, if $b \succ_a b'$. We let $E_1$ denote the set of edges that are not dominated by any edge in~$E$, we will call these \emph{first-choice edges}. 
We call an edge $e \in E$ {\em critical} if the maximum matching in~$G[E_1 \cup \{e\}]$ is larger than a maximum matching in $G[E_1]$. 
We define the set~$E_2$ of \emph{second-choice edges} as those edges that are critical, and are not dominated by any critical edge.

\begin{theorem}[\cite{AIKM07}]
\label{thm:characterization-AIKM07}
    A matching $M$ in $G = (A \cup B, E)$ is popular if and only if
    \begin{itemize}[leftmargin=24pt]
        \item[(1)]$M$ is a matching in $G[E_1 \cup E_2]$ that matches all agents in~$A$, and 
        \item[(2)] $M \cap E_1$ is a maximum matching in~$G[E_1]$.
    \end{itemize}
\end{theorem}

The above characterization was given in \cite{AIKM07} for instances where agents' preferences are weak rankings, and
the following example shows that it does not 
hold for general partial orders. Let $G=(A \cup B,E)$ where $A=\{a,b,c\}$ and $B=\{x,y,z\}$, each agent is adjacent to each object, and the preference orders of the agents are defined by 

$\begin{array}{lll}
    a: & x \succ_a z, & y \succ_a z; \\
    b: & x \succ_b z; \\
    c: & y \succ_c x, & y \succ_c z. \\
\end{array}$

\smallskip
Thus $x \sim_a y$, similarly, $x \sim_b y$ and $y \sim_b z$, and also $x \sim_c z$.
The set of first-choice edges is then $E_1=\{(
 a,x), (a,y), (b,x), (b,y), (c,y)\}$, the set of second-choice edges is $E_2=\{(a,z), (b,z), (c,z)\}$. Consider now the matching $M = \{(a,x), (b,y), (c,z)\}$. 
 It is straightforward to verify that $M$ satisfies properties~(1) and~(2) required in Theorem~\ref{thm:characterization-AIKM07}. 
However, $M$ is not popular, since the matching $N = \{(a,x), (b,z), (c,y)\}$ is more popular than~$M$, because $c$ prefers~$N$ to~$M$, while $a$ and~$b$ are indifferent between~$M$ and~$N$. 

\subsection{An instance without popular matchings that admits a popular assignment}\label{app:no-pop-assignment}

We describe a simple example that does not admit any popular matching, but admits a popular assignment. Let $G = (A \cup B, E)$ where $A = \{a_1,a_2,a_3\}$ and
$B = \{b_1,b_2,b_3\}$ and the preference order of both $a_1$ and $a_2$ is $b_1 \succ b_2$ while the preference order of $a_3$ is
$b_1 \succ b_2 \succ b_3$. 

It follows from the characterization of popular matchings from \cite{AIKM07} that a popular matching $M$ has to match
each of $a_1,a_2,a_3$ to either $b_1$ or $b_2$. Since this is not possible, this instance has no popular matching. It is easy to check that $M^* = \{(a_1,b_1),(a_2,b_2),(a_3,b_3)\}$ is a popular assignment in $G$.

\subsection{Some simple reductions to the popular assignment problem}\label{app:pop-matching-reduction}
We will first show a reduction from the \myproblem{popular matching} problem to the \myproblem{popular assignment} problem.
Let  $G = (A \cup B, E)$ be an instance of the \myproblem{popular matching} problem. 
Let $B' = B \cup \{l(a): a \in A\}$. That is, corresponding to each~$a \in A$, an object~$l(a)$ (the last resort of~$a$) 
is in~$B'$ and we set this object~$l(a)$ as the worst-choice of~$a$. Let $A' = A \cup \{d_1,\ldots,d_{|B|}\}$, i.e., there are $|B|$ many dummy agents in~$A'$.
Each dummy agent~$d_i$ is adjacent to all objects in~$B'$ and is indifferent between any two of them. 

It is easy to see that every matching~$M$ in~$G$ can be extended to a {\em perfect matching}~$M'$ in this new graph~$G' = (A'\cup B',E')$ and conversely, 
every perfect matching~$M'$ in~$G'$ projects to a matching $M$ in $G$. For any pair of matchings~$M$ and~$N$ in~$G$, observe that
$\Delta(M,N) = \Delta(M',N')$. Thus $M$ and $M'$ have the same unpopularity margin and, in particular, an algorithm that finds a popular assignment in~$G'$ solves the \myproblem{popular matching} problem in~$G$. 

\paragraph{\bf Popularity with diversity.} 
Recall this problem defined in Section~\ref{sec:intro} where every agent in an instance~$G = (A\cup B,E)$ 
has one of $k$ colors associated with it, and admissible matchings are those that for every $i \in \{1,\ldots,k\}$ match 
$c_i$ agents of color~$i$ so that $s_i \le c_i \le t_i$  for some given bounds~$s_i$ and~$t_i$. 
We seek a matching that is popular within the set of admissible matchings.

We augment $B$ by adding $n_i-s_i$ artificial objects for each $i$, where $n_i$ is the number of agents colored $i$.
For each $i$, these $n_i-s_i$ objects are tied as the worst-choices of all agents colored $i$. 
Let $A' = A \cup \{d_1,\ldots,d_{n'}\}$, where $n' = |B| - \sum_is_i$. 
Every dummy agent $d \in \{d_1, \dots, d_{n'}\}$ is adjacent to all objects in $B$ and for each color $i$ 
some fixed $t_i-s_i$ artificial objects meant for color class~$i$ 
introduced above---as before, $d$ is indifferent between any two of its neighbors.
So for each $i$, there are $n_i-t_i$ artificial objects not adjacent to any dummy agent. Let $G'$ be the new instance.
It is easy to see that an algorithm that finds a popular assignment in~$G'$ solves our problem in~$G$. 

\section{Dual certificates for popular assignments}
\label{sec:prelims}
Let $G = (A \cup B, E)$ be an input instance and let $\nu$ be the size of a maximum matching in $G$. Let us augment $G$ with $|B|- \nu$ dummy agents that are adjacent to all objects in $B$ (and indifferent among them), along with $|A|-\nu$ artificial objects that are tied as the worst-choice neighbors of all non-dummy agents. Any maximum matching~$M$ in 
the original graph extends to a perfect matching (i.e., assignment)~$M'$ in the augmented graph; moreover, $\Delta(M,N) = \Delta(M',N')$ for any pair of maximum matchings $M$ and~$N$ in $G$.
Thus, we can assume without loss of generality that the input instance~$G$ admits a perfect matching.

Let $|A| = |B| = n$ and $|E| = m$. Let $M$ be any perfect matching in $G$.
The following edge weight function~$\wt_M$ in $G$ will be useful. For any $(a,b) \in E$
\begin{equation*} 
\mathrm{let}\ \wt_M(a,b) = \begin{cases} \phantom{-} 1   & \text{if\ $a$\ prefers\ $b$\ to\ $M(a)$;}\\
	                     -1 &  \text{if\ $a$\ prefers\ $M(a)$\ to\ $b$;}\\			
                              \phantom{-} 0 & \text{otherwise,\ i.e., if\ $b \sim_a M(a)$.}
\end{cases}
\end{equation*}

Let $\wt_M(N)= \sum_{e \in N} \wt_M(e)$ for any edge set $N \subseteq E$.
Consider the following linear program~\ref{LP1} and its dual~\ref{LP2}.

\begin{minipage}[t]{0.4\linewidth}\centering
\begin{align}
 \label{LP1}
 \max \sum_{e \in E} & \wt_M(e)\cdot x_e  \tag{LP1}
\\   \notag
      \text{s.t.}\quad\sum_{e \in \delta(u)}x_e = 1  &\mbox{\hspace*{0.1in}}\forall\, u \in A \cup B\\ \notag
                        x_e  \ge 0   &\mbox{\hspace*{0.1in}}\forall\, e \in E. 
\end{align}
\end{minipage}
\hspace{0.4cm}
\begin{minipage}[t]{0.5\linewidth}\centering
  \begin{align}
  \label{LP2}
\min \sum_{u \in A \cup B}y_u  &\mbox{\hspace*{0.1in}}\tag{LP2}\\ \notag
       \text{s.t.}\quad y_{a} + y_{b} \ge \wt_{M}(a,b)  &\mbox{\hspace*{0.1in}}\forall\, (a,b)\in E. 
\end{align}
\end{minipage}

\bigskip
\ref{LP1} is well known to be integral, and hence its optimal value is $\max_N\wt_M(N)$ where $N$ is a perfect matching in $G$. The definition of $\wt_M$ implies 
that $\wt_M(N) = \Delta(N,M)$; recall the definition $\Delta(N,M) = \phi(N,M) - \phi(M,N)$.
So $M$ is a popular assignment if and only if the optimal value of~\ref{LP1} is at most 0. In fact, the optimal
value of \ref{LP1} is then exactly~0, by $\Delta(M,M) = 0$. Hence for a popular assignment $M$,
the edge incidence vector of $M$ is an optimal solution to~\ref{LP1}.

Theorem~\ref{thm:certificate} gives a useful characterization of popular assignments. 
The proof of Theorem~\ref{thm:certificate} is given in 
Section~\ref{sec:min-unpopular} along with the proof of a related result 
(Theorem~\ref{thm:dual-cert-unpopular}).

\begin{theorem}
\label{thm:certificate}
$M$ is a popular assignment if and only if there exists an optimal solution $\vec{\alpha}$ to~\ref{LP2}  such that
$\alpha_a \in \{0, 1, 2,\ldots, (n-1)\}$ for all $a \in A$,
$\alpha_b \in \{0, -1, -2,\ldots, -(n-1)\}$ for all $b \in B$, and $\sum_{u\in A\cup B}\alpha_u = 0$.
\end{theorem}

A {\em dual certificate} for a popular assignment~$M$ is an optimal  solution~$\vec{\alpha}$ to~\ref{LP2} satisfying the conditions given in Theorem~\ref{thm:certificate}.

\section{The popular assignment algorithm}
\label{sec:algo}

The goal of our algorithm is to construct a perfect matching $M$ in $G$ along with a dual certificate~$\vec{\alpha}$. 
Every $b \in B$ will have an associated {\em level} $\ell(b)$ in this algorithm and the $\alpha$-value
of $b$ will be $-\ell(b)$, 
i.e., we set $\alpha_b=-\ell(b)$.

Given a function $\ell:B \rightarrow \mathbb{N}$ called  a {\em level function}, for any $a \in A$ let $\ell^*(a)=\max_{b \in \Nbr(a)}\ell(b)$
be the highest level at which agent $a$ has neighbors.
Now we define the subgraph $G_\ell = (A \cup B, E_\ell)$ \emph{induced by levels $\ell(\cdot)$} by putting an edge $(a,b) \in E$ into $E_\ell$ if and only if 
\begin{itemize}[leftmargin=24pt]
\item[(i)] $b$ has level $\ell^*(a)$, and $a$ has no neighbor in level $\ell^*(a)$ that she prefers to $b$, or
\item[(ii)] $b$ has level $\ell^*(a)-1$, and $a$ prefers $b$ to each of her neighbors in level $\ell^*(a)$, and moreover, $a$ prefers
none of her neighbors in level $\ell^*(a)-1$ to $b$.
\end{itemize}
Thus in the subgraph $G_{\ell}$, 
every agent has edges to her favorite {\em highest-level} neighbors and to her favorite neighbors  one level below, provided these neighbors
are preferred to all of her highest-level neighbors (see Fig.~\ref{fig:G_ell} for an illustration). The following lemma will be very useful.

\begin{lemma}
  \label{lem:G-ell}
  A matching $M$ in $G$ is a popular assignment if and only if there exists a level function~$\ell$ such that $M$ is a perfect matching in
  $G_{\ell}$. Furthermore, this happens if and only if
  there is a level function~$\ell$ and 
  a dual certificate~$\vec \alpha$ for $M$ where
  $\ell(b) = |\alpha_b|$ for all $b \in B$ and $M$ is a perfect matching in $G_{\ell}$.
\end{lemma}
\begin{proof}
Let us first show that if there exists a level function $\ell$ such that $M$ is a perfect matching in $G_{\ell}$, then $M$ is a popular assignment in~$G$.  
  We construct a dual certificate for $M$ as follows.
  Let~$\alpha_b = -\ell(b)$ for all $b \in B$ and $\alpha_a = \ell(M(a))$ for all $a \in A$.
  Note that the value of $\vec{\alpha}$ as a solution for~\ref{LP2} is $\sum_{v \in A\cup B} \alpha_v = \sum_{a \in A} (\alpha_a + \alpha_{M(a)}) = 0 = \sum_{e \in M} \wt_M(e)$. Thus $\vec \alpha$ is optimal for~\ref{LP2} if it is feasible.
  It remains to show that $\alpha_a + \alpha_b \geq  \wt_M(a,b)$ for every $(a, b) \in E$. 
  
  So let $(a, b) \in E$ and let $b' := M(a)$. Note that $\alpha_a + \alpha_b = \ell(b') - \ell(b)$. We show that $\ell(b') - \ell(b) \geq \wt_M(a, b)$.
  Because $(a, b') \in E_\ell$, one of the following cases holds:
  \begin{itemize}
    \item Case (i): $\ell(b') = \ell^*(a)$, so $a$ prefers no neighbor of hers in level $\ell^*(a)$ to $b'$. We have two subcases:
    \begin{itemize}
      \item If $\ell(b) = \ell^*(a)$, then $a$ does not prefer $b$ to $b'$ and hence $\wt_M(a, b) \leq 0 = \ell(b') - \ell(b)$.
      \item If $\ell(b) < \ell^*(a)$, then $\ell(b') - \ell(b) \geq 1 \geq \wt_M(a, b)$.
    \end{itemize}
    \item Case (ii): $\ell(b') = \ell^*(a) - 1$, so $a$ prefers $b'$ to each of her neighbors in level $\ell^*(a)$, and $a$ prefers
    none of her neighbors in level $\ell^*(a)-1$ to $b'$.
    We have three subcases:
    \begin{itemize}
    \item If $\ell(b) = \ell^*(a)$, then $a$ prefers $b'$ to $b$ and hence $\wt_M(a, b) = -1 = \ell(b') - \ell(b)$.
    \item If $\ell(b) = \ell^*(a) - 1$, then $a$ does not prefer $b$ to $b'$ and hence $\wt_M(a, b) \le 0 = \ell(b') - \ell(b)$.
    \item If $\ell(b) < \ell^*(a) - 1$, then $\ell(b') - \ell(b) \geq 1 \geq \wt_M(a, b)$.
    \end{itemize}
  \end{itemize}
  Thus in each of these cases $\wt_M(a, b) \leq \ell(b') - \ell(b) = \alpha_a + \alpha_b$. Hence $\vec{\alpha}$ is a dual certificate for $M$, and thus $M$ is a popular assignment by Theorem~\ref{thm:certificate}.

  \smallskip

{
\begin{figure}\centering
\begin{minipage}{0.45\textwidth}
\centering
    \begin{tikzpicture}
    \draw [fill=black!10,draw=none,rounded corners] (1.7,1.8) rectangle (2.3,1.2);
    \draw [fill=black!10,draw=none,rounded corners] (1.7,0.8) rectangle (2.3,0.2);
    \draw [fill=black!10,draw=none,rounded corners] (1.7,-.2) rectangle (2.3,-1.4);
    \node[circle,fill=black,label=left:$a_1$,inner sep = 2pt] (a1) at (0,1){};
    \node[circle,fill=black,label=left:$a_2$,inner sep = 2pt] (a2) at (0,-1.1){};
    \node[rectangle, draw=black,inner sep = 2pt] (b1) at (2,1.5){$2$};
    \node[rectangle, draw=black,inner sep = 2pt] (b2) at (2,0.5){$1$};
    \node[rectangle, draw=black,inner sep = 2pt] (b3) at (2,-.5){$0$};
    \node[rectangle, draw=black,inner sep = 2pt] (b4) at (2,-1.1){$0$};
    \draw[very thick] (a1) -- (b1) node [fill=white, inner sep=0.05cm, pos=0.2] {\scriptsize $2$};
    \draw[dashed] (a1) -- (b2) node [fill=white, inner sep=0.05cm, pos=0.4]  {\scriptsize $3$};
    \draw[dashed] (a1) -- (b3) node [fill=white, inner sep=0.05cm, pos=0.2]  {\scriptsize $1$};
    \draw[very thick] (a2) -- (b2) node [fill=white, inner sep=0.05cm, pos=0.2]  {\scriptsize $2$};
    \draw[very thick] (a2) -- (b3) node [fill=white, inner sep=0.05cm, pos=0.4]  {\scriptsize $1$};
    \draw[very thick] (a2) -- (b4) node [fill=white, inner sep=0.05cm, pos=0.3]  {\scriptsize $1$};
    \node (d1) at (3.5,1.5) {};
    \end{tikzpicture}
    \caption{Illustration of the subgraph $G_{\ell}$ of~$G$ for an instance with weak rankings and a level function~$\ell$. Circles indicate agents and squares objects; the $\ell$-level of each object is written inside the square depicting it. Numbers on the edges indicate the agents' weak rankings. Bold edges are included in $G_{\ell}$ and dashed edges are not. All but two agents were omitted.}
    \label{fig:G_ell}
\end{minipage}\hspace{.3cm}
\begin{minipage}{0.45\textwidth}
    \centering
    \begin{tikzpicture}
    \node[circle,fill=black,label=left:$a_0$,inner sep = 2pt] (a1) at (0,3){};
    \node[circle,fill=black,label=left:$a_{t-1}$,inner sep = 2pt] (a2) at (0,2){};
    \node[circle,fill=black,label=left:$a_t$,inner sep = 2pt] (a3) at (0,1){};
    \node[circle,fill=black,inner sep = 2pt] (a4) at (0,0){};
    \node[rectangle, fill=black,inner sep = 2pt,label=right:$b_{0}$] (b1) at (2,3.5){};
    \node[rectangle, fill=black,inner sep = 2pt, label=right:$b_{t-1}$] (b2) at (2,2.5){};
    \node[rectangle, fill=black,inner sep = 2pt,label=right:$b_{t}$] (b3) at (2,1.5){};
    \node[rectangle, fill=black,inner sep = 2pt] (b4) at (2,.5){};
    \draw[thick,densely dashed] (a1) -- (b1); 
    \draw[thick] (a1) -- (b2) node[fill=white, inner sep=0.1cm, pos=0.5,rotate=-13]{$\dots$};
    \draw[thick, densely dashed] (a2) -- (b2);
    \draw[thick] (a2) -- (b3);
    \draw[thick,densely dashed] (a3) -- (b3) node[pos=0.7,below]{\scriptsize $\not \in E_{\ell}$};
    \draw[thick] (a3) -- (b4) node[fill=white, inner sep=0.1cm, pos=0.5,rotate=-13]{$\dots$};
    \draw[thick, densely dashed] (a4) -- (b4);
    \end{tikzpicture}
    \caption{An illustration of the $M$-augmenting path~$P$ within the proof of Lemma~\ref{lem:alpha-vs-ell}. Solid edges are in~$M$ and dashed edges are in $M^\star$. The edge $(a_t, b_t)$ is not contained in $E_\ell$.} 
    \label{fig:my_label2}
\end{minipage}
\end{figure}
}

  We will now show the converse. Let $M$ be a popular assignment in $G$ and let $\vec{\alpha}$ be a dual certificate for $M$.
   We claim that $M$ is a matching in the graph $G_{\ell_{\vec\alpha}}$ induced by levels $\ell_{\vec{\alpha}}$ with $\ell_{\vec{\alpha}}(b) = |\alpha_b|$ for all $b \in B$.
  To prove this, we use that $\alpha_a + \alpha_b \ge \wt_{M}(a,b)$ for every $(a,b) \in E$.  
  First, because the incidence vector of $M$ and $\vec \alpha$ are optimal solutions to~\ref{LP1} and~\ref{LP2}, respectively, we get $\alpha_a + \alpha_{M(a)} = 0$ for each $a \in A$ by complementary slackness. This implies $\alpha_a = -\alpha_{M(a)} = \ell_{\vec \alpha}(M(a))$.
  Therefore, $\ell_{\vec \alpha}(M(a)) \geq \ell_{\vec \alpha}(b) + \wt_M(a, b)$ for all $(a, b) \in E$.  
  
  Since $\wt_{M}(e) \ge -1$ for all edges $e$, any agent $a$ has to be matched in~$M$ to either 
  (i)~an {\em undominated} neighbor in level $\ell^*_{\vec\alpha}(a)$ (i.e., $a$ prefers none of her neighbors in this level to~$M(a)$)
  or (ii)~an undominated neighbor in level $\ell^*_{\vec\alpha}(a)-1$ which, moreover, has to {\em dominate} (i.e., be preferred by $a$ to) all of $a$'s neighbors in level~$\ell^*_{\vec\alpha}(a)$. So $M$ is a perfect matching in $G_{\ell_{\vec\alpha}}$. \qed
\end{proof}

\paragraph{\bf The algorithm.}
Consider Algorithm~\ref{alg:popassign} on input $G = (A \cup B, E)$.
In search for a dual certificate, this algorithm will maintain a \emph{level} $\ell(b)$ for every $b \in B$.
Initially, $\ell(b) = 0$ for every $b \in B$.

The algorithm checks whether there exists 
a popular assignment by computing a perfect matching in the graph $G_\ell$.
If no such matching exists, the levels of unmatched objects are increased,  the graph~$G_{\ell}$ is updated accordingly, and the search continues. 

Eventually, either a perfect matching in $G_\ell$ is found,
or the level of an object exceeds $n-1$. In the latter case we can conclude that no popular assignment exists, as we will show below.

\begin{algorithm}
\caption{Finding a popular assignment in $G = (A\cup B, E)$}\label{alg:popassign}
\begin{algorithmic}[1]
\ForAll{$b \in B$} $\ell(b) = 0$. \EndFor 
\While{$\ell(b) < n$ for all $b \in B$}
\State Construct the graph $G_\ell$ and compute a maximum matching $M$ in $G_{\ell}$. 
\If{$M$ is a perfect matching} \textbf{return} $M$. \EndIf
\ForAll{$b \in B$ unmatched in $M$}  $\ell(b) = \ell(b) + 1$. \EndFor
\EndWhile
\State \textbf{return} ``$G$ has no popular assignment''.  
\end{algorithmic}
\end{algorithm}

\paragraph{\bf Running time.} Computing a maximum matching in $G_\ell$ takes $O(m\sqrt{n})$ time.
In every iteration of the algorithm, the value $\sum_{b\in B}\ell(b)$ increases.
So the number of iterations is at most $n^2$. Hence the running time of our algorithm is $O(m\cdot n^{5/2})$.

\begin{theorem}
  \label{thm:pop}
  If our algorithm returns a matching $M$, then $M$ is a popular assignment in $G$.
\end{theorem}

Theorem~\ref{thm:pop} follows immediately from Lemma~\ref{lem:G-ell}. The more difficult part in our proof of correctness is to show that whenever our algorithm says that $G$ has no popular assignment, the instance~$G$ indeed has no popular assignment. This is implied by Theorem~\ref{thm:pop2}.

\begin{theorem}
  \label{thm:pop2}
  Let $M^\star$ be a popular assignment in $G$ and let $\vec{\alpha}$ be a dual certificate of $M^\star$.
  Then for every $b \in B$, we have $|\alpha_b| \ge \ell(b)$, where $\ell(b)$ is the level of $b$ when our algorithm terminates.
\end{theorem}

If our algorithm terminates because $\ell(b) = n$ for some $b \in B$, then $|\alpha_b| \ge n$ for any dual certificate~$\vec{\alpha}$ by Theorem~\ref{thm:pop2}. However  $|\alpha_b| \le n-1$ by definition, a contradiction. So $G$ has no popular assignment.

 The following lemma is crucial for proving Theorem~\ref{thm:pop2}. It guarantees that when the algorithm increases $\ell(b)$ for some unmatched object $b \in B$, then the new level function does not exceed $|\alpha_b|$.  

\begin{lemma}\label{lem:alpha-vs-ell}
Let $M^\star$ be a popular assignment, let $\vec \alpha$ be a dual certificate of $M^\star$, and let  
$\ell: B \rightarrow \mathbb{N}$ be such that $\ell(b) \leq |\alpha_b|$ for all $b \in B$.
Let $M$ be a maximum matching in $G_\ell$ and let $b_0 \in B$ be an object that is 
left unmatched in $M$.
Then $\ell(b_0) < |\alpha_{b_0}|$.
\end{lemma}

 Before we turn to the proof of Lemma~\ref{lem:alpha-vs-ell}, we point out that Theorem~\ref{thm:pop2} follows from this lemma by a simple induction.

\begin{proof}[of Theorem~\ref{thm:pop2}]
  Let $\ell_i$ be the level function on the set $B$ at the start of the $i$-th iteration of our
algorithm. 
  We are going to show by induction that for every~$i$ we have $|\alpha_b| \ge \ell_i(b)$ for all~$b \in B$.
  This is true for $i = 1$, since $\ell_1(b) = 0$ for all~$b \in B$.
  Now suppose that $|\alpha_b| \ge \ell_i(b)$ for all $b \in B$.
  Let $b_0 \in B$.
  If $b_0$ is matched in the maximum matching $M$ in $G_{\ell_i}$, then we know $\ell_{i+1}(b_0) = \ell_{i}(b_0) \leq |\alpha_{b_0}|$. If $b_0$ is 
  left unmatched in $M$, then $\ell_{i+1}(b_0) = \ell_{i}(b_0) + 1$ and $\ell_{i}(b_0) < |\alpha_{b_0}|$ by Lemma~\ref{lem:alpha-vs-ell}. Thus $\ell_{i+1}(b_0) \leq |\alpha_{b_0}|$ in either case, completing the induction. \qed
\end{proof}

\begin{proof}[of Lemma~\ref{lem:alpha-vs-ell}]
  By Lemma~\ref{lem:G-ell}, $M^\star$ is a perfect matching in  $G_{\ell_{\vec\alpha}}$ where $\ell_{\vec\alpha}(b) = |\alpha_b|$ for~$b \in B$.
  Thus, the symmetric difference $M \oplus M^\star$ in  $G$ contains an $M$-augmenting path~$P$ starting at~$b_0$.
  However, as $M$ is of maximum size in~$G_{\ell}$, the path $P$ must contain an edge that is not in~$E_{\ell}$.

  Let $(b_0, a_0, b_1, a_1, \dots, b_t, a_t)$ be any prefix of $P$ such that $(a_t, b_t) \notin E_{\ell}$ (see Fig.~\ref{fig:my_label2}). 
  Note that $(a_0, b_0)$ and $(a_t, b_t)$ are in $M^\star$, since $M$ leaves $b_0$ unmatched and $M \subseteq E_{\ell}$. 
  Thus $(a_h, b_h) \in M^{\star} \subseteq E_{\ell_{\vec\alpha}}$ for all $h \in \{0, \dots, t\}$ and $(a_h, b_{h+1}) \in M \subseteq E_{\ell}$ for all $h \in \{0, \dots, t - 1\}$.
  We will show that $\ell(b_h) < \ell_{\vec\alpha}(b_h)$ for all $h \in \{0, \dots, t\}$, and thus in particular, $\ell(b_0) < \ell_{\vec\alpha}(b_0) = |\alpha_{b_0}|$.

  We first show that $\ell(b_t) < \ell_{\vec\alpha}(b_t)$.
  Assume for contradiction that $\ell(b_t) = \ell_{\vec\alpha}(b_t)$.
  Using the fact that $(a_t, b_t) \notin E_{\ell}$, one of the following cases must hold:
  \begin{itemize}
    \item[$\bullet$] $a_t$ has a neighbor in level at least $\ell(b_t)+2$, or
    \item[$\bullet$] $a_t$ has a neighbor in level $\ell(b_t)+1$ that is not dominated by $b_t$,  or
    \item[$\bullet$] $a_t$ has a neighbor in level~$\ell(b_t)$ that is preferred to $b_t$.
  \end{itemize}  
  As $\ell_{\vec\alpha}(b_t) = \ell(b_t)$ and $\ell_{\vec\alpha}(b) \ge \ell(b)$ for all $b \in B$, 
  in each case we get $(a_t, b_t) \notin E_{\ell_{\vec\alpha}}$, a contradiction.
  
  Now suppose there exists some $h \in \{0, \dots, t-1\}$ with $\ell(b_{h+1}) < \ell_{\vec\alpha}(b_{h+1})$ but $\ell(b_h)= \ell_{\vec\alpha}({b_h})$. Recall that $(a_h,b_{h+1}) \in M \subseteq E_\ell$, which leaves us with the following possibilities: 
\begin{itemize}
\item[$\bullet$] $\ell(b_{h+1}) \ge \ell(b_h) + 1$: then $\ell_{\vec\alpha}(b_{h+1}) \geq \ell_{\vec\alpha}(b_h) + 2$;
\item[$\bullet$] $\ell(b_{h+1}) = \ell(b_h)$: then $a_h$ does not prefer $b_h$ to $b_{h+1}$, but $\ell_{\vec\alpha}(b_{h+1}) \ge \ell_{\vec\alpha}(b_{h}) + 1$;
\item[$\bullet$] $\ell(b_{h+1}) = \ell(b_h) - 1$: then $a_h$ prefers $b_{h+1}$ to $b_h$, but $\ell_{\vec\alpha}(b_{h+1}) \ge \ell_{\vec\alpha}(b_{h})$.
\end{itemize} 
In each of these cases,  we get $(a_h,b_h) \notin E_{\ell_{\vec\alpha}}$ by the definition of $G_{\ell_{\vec\alpha}}$, again a contradiction. \qed
\end{proof}

We remark that if a popular assignment exists, then the algorithm returns a popular assignment~$M$ and a corresponding dual certificate $\vec\alpha$ such that $\ell_{\vec\alpha} \leq \ell_{\vec\alpha'}$ for any dual certificate $\vec\alpha'$ of any popular assignment $M'$. This shows that there is a unique minimal dual certificate in this sense. 

\smallskip

We will show that a two-level truncation of Algorithm~\ref{alg:popassign} (i.e., where the \textsf{while}-loop terminates if $\ell(b) = 2$ for some $b\in B$)
solves the \myproblem{popular matching} problem with partial order preferences. Before we do this, we point out
a generalization of Lemma~\ref{lem:alpha-vs-ell} that encapsulates the main argument of the preceding proof. This insight will be useful for generalizing our algorithmic result in Sections~\ref{sec:edge-restrictions}-\ref{sec:min-unpopular}.

\begin{lemma}\label{lem:alpha-vs-ell-general}
Let $\ell, \ell': B \rightarrow \mathbb{N}$ be such that $\ell(b) \leq \ell'(b)$ for all $b \in B$.
Let $M$ and $M'$ be matchings in~$G_\ell$ and~$G_{\ell'}$, respectively.
Let $b_0 \in B$ be an object that is matched in $M'$ but not in $M$.
Let $P$ be the path in $M \oplus M'$ containing~$b_0$.
If $P$ contains an edge not in $E_\ell$,
then $\ell(b_0) < \ell'(b_0)$.
\end{lemma}

\subsection{Implications for the \myproblem{popular matching} problem}
\label{sec:2-level-algo}
Recall the reduction from the \myproblem{popular matching} problem to the \myproblem{popular assignment} problem, as described in Section~\ref{app:pop-matching-reduction}.
It can be shown that for graphs obtained 
through this reduction
(i.e., those with last resort objects and dummy agents)
we can terminate our algorithm as soon as some object reaches level $2$. 
This is because in any such instance, there exists a dual certificate $\vec{\alpha}$ for a popular assignment that satisfies $|\alpha_b| \leq 1$ for all $b \in B$, as shown in Proposition~\ref{prop:alg-level2} below.

\begin{proposition}
\label{prop:alg-level2}
Let $G=(A \cup B, E)$ be an instance of the \myproblem{popular matching} problem, and let $G'=(A' \cup B',E')$ be the corresponding instance of \myproblem{popular assignment} as constructed in Section~\ref{app:pop-matching-reduction}. If a popular matching exists in~$G$, then there exists a dual certificate $\vec{\alpha}$ for a popular assignment~$M'$ in~$G'$ that satisfies $|\alpha_{b}| \leq 1$ for all $b \in B'$.
\end{proposition}
\begin{proof}
It is straightforward to check that the reduction in  Section~\ref{app:pop-matching-reduction} is correct, so $G$ admits a popular matching if and only if $G'$ admits a popular assignment. 
If $M'$ is a popular assignment in~$G'$, then by Theorem~\ref{thm:certificate} there exists a dual certificate~$\vec{\alpha}$ for~$M'$.

We may assume that there exists some  object $b_0 \in B'$ for which $\alpha_{b_0}=0$.
Otherwise we have $\alpha_b \leq -1$ for each $b \in B'$, and since $\alpha_b+\alpha_{M'(b)}=0$ for each $b \in B'$ by complementary slackness, we also have  $\alpha_a \geq 1$ for each $a \in A'$. 
Therefore, decreasing the $\alpha$-value of each agent by~$\min_{a \in A'} \alpha_a$ while increasing the $\alpha$-value of each object by the same amount results in a dual certificate where some object~$b_0$ has $\alpha$-value~$0$.

Assume now for the sake of contradiction that $\alpha_b \leq -2$ for some~$b \in B'$, and consider the agent~$M'(b_0)$. Recall that $\alpha_{M'(b_0)}=0$ follows from $\alpha_{b_0}=0$.

On the one hand, if $M'(b_0)$ is a dummy agent~$d$, then by construction it is adjacent to~$b$, and moreover, $\wt_{M'}(d,b)=0$, because $d$ is indifferent between all objects in~$B'$. However, then $\alpha_d+\alpha_b \leq-2 < \wt_{M'}(d,b)=0$ contradicts the feasibility of~$\vec{\alpha}$.

On the other hand, if $M'(b_0)$ is a non-dummy agent $a \in A$, then $\alpha_{l(a)} \geq -1$ follows for the last resort object $l(a)$ of~$a$, due to $\wt_{M'}(a,l(a)) \geq -1$.
Since $a$ is the unique non-dummy neighbor of~$l(a)$ in~$G'$ and $M'(a)=b_0$, we know that $l(a)$ must be assigned to some dummy agent~$d'$ in~$M'$. 
By $\alpha_{l(a)} \geq -1$ we know $\alpha_{d'} \leq 1$; however, then the edge $(d',b)$ contradicts the feasibility of~$\vec{\alpha}$ by~$\alpha_{d'}+\alpha_b \leq -1<\wt_{M'}(d',b)=0$. \qed
\end{proof}

Thanks to Proposition~\ref{prop:alg-level2}, we can solve the \myproblem{popular matching} problem by using Algorithm~\ref{alg:popassign} and terminating it with rejection as soon as some object reaches level~$2$. This implies that the number of iterations the algorithm performs is at most the number of objects.  
Thus we have the following corollary.

\begin{corollary}
The \myproblem{popular matching} problem where agents' preferences are partial orders can be solved in $O(m \cdot n^{3/2})$ time on a graph with $n$ vertices and $m$ edges.
\end{corollary}

Note that for weak rankings, a two-level truncation of Algorithm~\ref{alg:popassign}
essentially leads to the popular matching algorithm by Abraham et al.~\cite{AIKM07}.
However these two algorithms are quite different for partial order preferences, e.g., consider the example in
Section~\ref{app:characterization-of-pop-matchings}. Running Algorithm~\ref{alg:popassign} on this instance, we obtain $\ell(x) = \ell(y) = 0$ and $\ell(z) = 1$ at the last iteration. 
This implies $(b,y)\notin E_\ell$, hence the unpopular matching $\{(a,x),(b,y),(c,z)\}$ is not contained in $E_\ell$.

\subsection{Popularity with penalties}
\label{sec:k-level-algo}

In the following, we describe a family of problems that interpolates between the \myproblem{popular matching} and \myproblem{popular assignment} problem: given an instance $G=(A\cup B, E)$ of \myproblem{popular matching}, we say that a matching~$M$ in~$G$ is \emph{popular with penalty~$\k$} for some positive integer $\k$, if there is no matching~$N$ in~$G$ such that $N$ beats $M$ in an election where agents vote between~$M$ and~$N$ according to their preferences, but votes of unmatched agents are counted with weight~$\k$. Formally, given two matchings~$N$ and~$M$, define for any agent~$a$
\[
\vote^\k_a(N,M)
= \begin{cases} \phantom{-} 1   & \text{if\ $a$\ is matched in~$N$ and~$M$, and prefers\ $N(a)$\ to\ $M(a)$;}\\
	                     -1 &  \text{if\ $a$\ is matched in~$N$ and~$M$, and prefers\ $M(a)$\ to\ $N(a)$;}\\			
                      \phantom{-} \k &  \text{if\ $a$\ is matched in~$N$ but not in~$M$;}\\		
                      -\k &  \text{if\ $a$\ is matched in~$M$ but not in~$N$;}\\		
                              \phantom{-} 0 & \text{otherwise.}
\end{cases}
\]
Then, the matching~$M$ is \emph{popular with penalty~$\k$}, if $\sum_{a \in A} \vote^\k_a(N,M) \leq 0$ for each matching~$N$. Clearly, for $\k=1$, this property corresponds to the notion of popularity for matchings. Moreover, as long as $G$ contains an assignment (perfect matching), a matching $M$ is popular with penalty $\k=n$ if and only if $M$ is a popular assignment. More generally, any augmenting path $P$ with respect to a matching $M$ that is popular with penalty~$\k$ contains at least $\k+1$ agents, otherwise  $\sum_{a\in A}\vote^\k_a(M\oplus P,M) > 0$. Hence $|M| \ge \frac{\k}{\k+1}|M_{\max}|$ where $M_{\max}$ is a maximum matching in $G$.

In the following we show that a matching that is popular with penalty $\k$ can be found via a reduction to an instance of the \myproblem{popular assignment} problem.
Moreover, we also show that we can determine whether there is an assignment that is popular with penalty $\k$ (when compared with matchings from $G$) by running a truncated version of Algorithm~\ref{alg:popassign}.

\subsubsection {Reduction for finding a matching that is popular with penalty~$\k$.}
Given an instance $G=(A \cup B,E)$ of \myproblem{popular matching} and an integer~$\k$, let us construct an instance of \myproblem{popular assignment}~$G' = (A' \cup B', E')$ as follows: 
\begin{itemize}
    \item The set $A'$ of agents 
contains all agents in $A$, a set~$D=\{d_1,\dots,d_{|B|}\}$ of $|B|$ dummy agents, and a set $\{p_1(a),\dots, p_{\k-1}(a)\}$ of new agents for each $a \in A$.
    \item  The set $B'$ of objects contains all objects in $B$ and a set~$\{l_1(a), \dots, l_\k(a)\}$ of new objects for each $a \in A$.
\end{itemize}
The edge set $E'$ of $G'$ contains the following edges in addition to the ones in the original set of edges $E$: 
\begin{itemize}
    \item For each $a \in A$, a path $P(a)=(a,l_1(a),p_1(a),l_2(a),\dots, p_{\k-1}(a),l_\k(a))$ is added; 
    moreover, the dummy agents are adjacent to every object in $B \cup L_\k$ where $L_\k=\{l_\k(a):a \in A\}$. 
\end{itemize}

The preferences of the original instance are extended as follows: 
Each original agent $a \in A$ keeps her original preference order and in addition, prefers any object $b \in B$ with $(a, b) \in E$ to object~$l_1(a)$, i.e., $l_1(a)$ is $a$'s unique worst-choice object.
For $a \in A$ and $i \in [\k-1]$, agent $p_i(a)$ prefers object~$l_{i}(a)$ to object $l_{i+1}(a)$.
Finally, each dummy agent $d \in D$ is indifferent among her adjacent objects.

For a matching $M$ in $G$, we define a \emph{corresponding assignment} $M'$ in $G'$ as follows.
For each agent $a \in A$ that is matched in $M$, define $M'(a) = M(a)$ and $M'(p_i(a)) = l_{i}(a)$ for $i \in [\k-1]$.
For each agent $a \in A$ that is left unmatched in~$M$, define $M'(a) = l_1(a)$ and  $M'(p_i(a)) = l_{i+1}(a)$ for $i \in [\k-1]$.
Finally, we add an arbitrary perfect matching between the dummy agents in $D$ and the~$|B|$ objects in $\{b \in B : b \text{ is left unmatched in } M\} \cup 
\{l_\k(a): a \text{ is matched in } M\}$.

Note that, by construction of $G'$, for every assignment $N'$ in $G'$ and every $a \in A$, either $N'$ contains the edges $(a, l_1(a))$, $(p_1(a), l_2(a))$, \dots, $(p_{\k-1}(a), l_\k(a))$ or we have $N'(a) \in B$, $N'(l_\k(a)) \in D$, and $N'$ contains the edges $(p_1(a), l_1(a))$, \dots, $(p_{\k-1}(a), l_{\k-1}(a))$.
Hence, up to permutation of the dummy agents (which, due their indifference, do not influence pairwise comparisons) there is a one-to-one correspondence between matchings in $G$ and assignments in $G'$.

\begin{proposition}
\label{prop:vote-for-penalty}
Let $M$ and $N$ be two matchings in $G$ and let $M'$ and $N'$ be the corresponding assignments in $G'$. Then  $\Delta(N',M')=\sum_{a \in A}  \vote^\k_a(N,M)$.
\end{proposition}
\begin{proof}
When comparing assignments~$N'$ and~$M'$, we write $\vote_a(N',M')$ instead of $\vote^\k_a(N',M')$, since no agent is unmatched in an assignment. We claim that 
\begin{equation}
\label{eq:vote}
    \vote_{a}(N', M') + \sum_{i = 1}^{\k-1} \vote_{p_i(a)}(N', M') = \vote^\k_{a}(N, M)
\end{equation}
for each agent~$a \in A$. 
Note that for every agent $a \in A$ that is either matched in both $M$ and $N$ or unmatched in both $M$ and $N$, every agent $p_i(a)$ for $i \in [\k-1]$ is matched to the same object in both $M'$ and $N'$.
Hence \eqref{eq:vote} holds for such agent~$a$.
Now consider an agent $a \in A$ that is matched in $N$ but not in $M$.
Note that in this case $\vote_a(N', M') + \sum_{i = 1}^{\k-1} \vote_{p_i(a)}(N', M') = \k$, because $\vote_a(N', M') = 1$ and because $p_i(a)$  prefers $N(p_i(a)) = l_i(a)$ to $M(p_i) = l_{i+1}(a)$ for each $i \in [\k-1]$, 
so (\ref{eq:vote}) holds again.
By a symmetric argument, $\vote_a(N', M') + \sum_{i = 1}^{\k-1} \vote_{p_i(a)}(N', M') = -\k$ for all agents $a \in A$ that are unmatched in $N'$ but matched in $M'$. This proves our claim.
Because $\vote_d(N', M') = 0$ for all $d \in D$, we conclude that 
\begin{align*}
    \Delta(N', M') & = \sum_{a \in A'} \vote_a(N', M') = \sum_{a \in A} \left(\vote_a(N', M') + \sum_{i = 1}^{\k-1} \vote_{p_i(a)}(N', M')\right) \\
    & = \sum_{a \in A} \vote^\k_a(N, M). 
\end{align*}    
This completes the proof. \qed
\end{proof}

By Proposition~\ref{prop:vote-for-penalty}, a matching $M$ in $G$ is popular with penalty $\k$ in $G$ if and only if the corresponding assignment $M'$ in $G'$ is popular in $G'$.
In particular, we can find a matching in $G$ that is popular with penalty $\k$, if it exists, in polynomial time by applying Algorithm~\ref{alg:popassign}. 
Furthermore, the structure of $G'$ allows us to show that the \emph{$(\k+1)$-level truncation} of Algorithm~\ref{alg:popassign} suffices to solve
the \myproblem{popular assignment} problem in $G'$, i.e.,
we terminate the algorithm with rejection as soon as some object reaches level~$\k+1$.
The proof of the following proposition is given in the appendix.

\begin{restatable}{proposition}{proppenalty}
\label{prop:penalty-truncation}
    If $M'$ is a popular assignment in~$G'$, then there exists a dual certificate~$\vec{\alpha}$ for~$M'$ such that $|\alpha_{b}| \leq \k$ for each object $b  \in B'$.
\end{restatable}

Recall that there are $O(n\k)$ nodes and $O(n^2)$ edges in~$G'$ where $n=|A \cup B|$ and $m=|E|$. 
Thus, the following result follows immediately from Proposition~\ref{prop:vote-for-penalty} and Proposition~\ref{prop:penalty-truncation}.

 \begin{theorem}
    Given an instance $G = (A \cup B, E)$ of the \myproblem{popular matching} problem where $|A\cup B| = n$ along with $\k \in [n-1]$,  
    a matching that is popular with penalty~$\k$ can be computed in $O(n^{7/2} \cdot \k^{5/2})$ time, if $G$ admits such a matching.
\end{theorem}    

\subsubsection{Finding a popular assignment with penalty $\k$.} 
In the following, we show that it is possible to determine whether an instance of \myproblem{popular matching with penalty $\k$} has a solution $M$ that is an assignment.
For any assignment $N$, we have $\sum_{a \in A} \vote^\k_a(N,M)=\Delta(N,M)$, hence any such solution $M$ is also a popular assignment. Moreover, 
for popularity with penalty $\k$, an assignment is also compared with matchings in $G$.
This means that for an assignment, popularity with a given penalty is a stronger requirement than popularity within the set of assignments.

Below, we show that the $(\k + 1)$-level truncation of Algorithm~\ref{alg:popassign} finds an assignment that is popular with penalty $\k$ (if it exists). In contrast to the previous result, we can directly execute the truncated algorithm on $G$ without applying the reduction.

\begin{theorem}
\label{thm:pop-w-penalty}
    For any positive integer~$\k$, the $(\k+1)$-level truncation of Algorithm~\ref{alg:popassign} finds an assignment that is popular with penalty~$\k$, if such an assignment exists.
    The running time on a graph with $n$ vertices and $m$ edges is $O(\k \cdot m \cdot n^{3/2})$.
\end{theorem}
\begin{proof}
Suppose first that Algorithm~\ref{alg:popassign} outputs an assignment~$M$ while all objects have level at most~$\k$; by the correctness of Algorithm~\ref{alg:popassign} we know that $M$ is a popular assignment. We claim that $M$ is also popular with penalty~$\k$. 

For the sake of contradiction, assume that $N$ is a matching that beats~$M$, i.e., for which $\sum_{a \in A} \vote^\k_a(N,M)>0$, and consider the symmetric difference~$N \oplus M$ of~$N$ and~$M$. 
For each connected component~$K$ in~$N \oplus M$, let $A_K$ and $B_K$ denote the set of agents and objects in~$K$, respectively.
Observe that there must exist a connected component~$K$ of $N \oplus M$ for which we have $\sum_{a \in A_K} \vote^\k_a(N,M) > 0$. 

First, if $K$ is a cycle, then the matching~$M'=M \oplus K$ is an assignment, and we get that $\sum_{a \in A} \vote^\k_a(M',M)=\sum_{a \in A_K} \vote^\k_a(M',M) = \sum_{a \in A_K} \vote^\k_a(N,M)>0$, which contradicts the fact that $M$ is a popular assignment. 

Second, if $K$ is a path, then since $M$ is a perfect matching, the two endpoints of~$K$ must be an agent~$a$ and an object~$b$, both matched in~$M$ but unmatched in~$N$. 
Let the vertices of the path~$K$ be $a=a_0,b_0,a_1,b_1,\dots,a_t,b_t=b$ for some $t \in \mathbb{N}$, so that $M(a_i)=b_i$ for $i=0,1,\dots,t$ and $N(a_i)=b_{i-1}$ for $i=1,\dots,t$. 
Let $\vec{\alpha}$ be the dual certificate for~$M$
for which $\ell_{\vec{\alpha}}$ is the level-function obtained by Algorithm~\ref{alg:popassign}. 

On the one hand, by complementary slackness we know $\sum_{i=0}^t (\alpha_{a_i}+\alpha_{b_i})=0$. 
On the other hand, for each $i \in [t]$ we know $\wt_M(a_i,b_{i-1})= \vote^\k_{a_i}(N,M)$ by definition, which implies 
\begin{align*}
\sum_{i=1}^t (\alpha_{a_i}+  \alpha_{b_{i-1}}) 
\geq & \sum_{i=1}^t \wt_M(a_i,b_{i-1}) 
= \sum_{i=1}^t  \vote^\k_{a_i}(N,M) \\
= & \left( \sum_{i=0}^t  \vote^\k_{a_i}(N,M) \right)
-  \vote^\k_{a_0}(N,M) = \left( \sum_{i=0}^t  \vote^\k_{a_i}(N,M) \right) + \k > \k
\end{align*}
because $ \vote^\k_{a_0}(N,M)=-\k$.
This leads to 
\[
0=
\sum_{i=0}^t (\alpha_{a_i}+\alpha_{b_{i}}) 
= \alpha_{a_0}+ \alpha_{b_t} +
\sum_{i=1}^t 
\left(
\alpha_{a_i}+ \alpha_{b_{i-1}}
\right) 
> \alpha_{b_t} +\k
\]
by $\alpha_{a_0} \geq 0$. However, this means $\ell_{\vec{\alpha}}(b_t)=|\alpha_{b_t}| >\k$ which contradicts our assumption that Algorithm~\ref{alg:popassign} returned $M$ while all objects had level at most~$\k$.
This proves that $M$ is indeed popular with penalty~$\k$.

\smallskip
For the other direction, we show that if $M$ is an assignment that is popular with penalty~$\k$, then it admits a dual certificate~$\vec{\alpha}$ such that $\alpha_b \geq -\k$ for each object $b \in B$. Since Algorithm~\ref{alg:popassign} returns a level function $\ell$ satisfying $\ell(b) \leq \ell_{\vec{\alpha}}(b)=|\alpha_b|$ for each~$b \in B$, this means that the $(\k+1)$-level truncation of Algorithm~\ref{alg:popassign} will return an assignment as desired.

Recall first that $M$ is a popular assignment, so let~$\vec{\alpha}$ be a dual certificate for~$M$ minimizing $\sum_{a \in A} \alpha_a$.
Assume for the sake of contradiction
that $\alpha_{b^\star}<-\k$ for some 
$b^\star \in B$.
Let us call an edge~$(a,b)$ in~$G$ \emph{tight}, if $\alpha_a+\alpha_b=\wt_M(a,b)$; clearly, all edges of~$M$ are tight by complementary slackness.
Let us consider the set  $A^\star$ of agents and the set $B^\star$ of objects in~$G$ that are reachable from~$b^\star$ via an $M$-alternating path in~$G$ consisting only of tight edges and starting with the edge~$(M(b^\star),b^\star)$; for convenience, let us add~$b^\star$ to~$B^\star$.

We show that there exists some $a^\star \in A^\star$ with $\alpha_{a^\star}=0$. If not, then consider the modification~$\vec{\beta}$ of~$\vec{\alpha}$ obtained through increasing the $\alpha$-value of each object in~$B^\star$ by~$1$ and decreasing the $\alpha$-value of each agent in~$A^\star$ by~$1$. We claim that $\vec{\beta}$ is a dual certificate. 
First note that an agent~$a$ is contained in~$A^\star$ if and only if $M(a)$ is contained in~$B^\star $, by construction.
Hence, as no agent in $A^\star$ has level~$0$, the $\beta$-value of each agent remains non-negative, and the $\beta$-value of each object remains non-positive. 
Observe also that $\vec{\beta}$ is a feasible solution for~\ref{LP2}: assuming that \ref{LP2} is violated on some edge~$(a,b)$, it must be the case that $a \in A^\star$ but $b \notin B^\star$ and $(a,b)$ is tight, otherwise $\beta_a+\beta_b \geq \wt_M(a,b)$ still holds, since $\vec{\alpha}$ is a feasible solution for~\ref{LP2}. However, then the edge~$(a,b)$ is not in~$M$ but is tight, which implies that appending $(a,b)$ to an $M$-alternating path from~$b^\star$ to~$a$ witnessing~$a \in A^\star$ yields a path that shows $b \in B^\star$, a contradiction. Hence, $\vec{\beta}$ is also a dual certificate, contradicting the minimality of~$\vec{\alpha}$. This proves that there exists some $a^\star \in A^\star$ with $\alpha_{a^\star}=0$. 

Now, let $P$ be an $M$-alternating path of tight edges between $b^\star$ and $a^\star$ that 
witnesses $a^\star \in A^\star$,
and let $A_P$ and~$B_P$ denote the agents and objects appearing on~$P$, respectively. 
Consider the matching $N=M \oplus P$ in which $a^\star$ and $b^\star$ are unmatched, and observe that
\begin{align*}
\sum_{a \in A}  \vote^\k_a   (N,M)
 & = \sum_{a \in A_P}  \vote^\k_a(N,M)
=  \vote^\k_{a^\star}(N,M)+\sum_{a \in A_P \setminus \{a^\star\}}  \vote^\k_a(N,M)  \\
& = -\k + \sum_{a \in A_P\setminus \{a^\star\}} \wt_M(a,N(a))=
-\k + \sum_{a \in A_P\setminus \{a^\star\}} \alpha_a + 
\sum_{b \in B_P \setminus \{b^\star\}} \alpha_b  \\
& = 
-\k + \left(\sum_{a \in A_P}\alpha_a + \sum_{B \in B_P}\alpha_b \right) - \alpha_{a^\star} - \alpha_{b^\star} = -\k +0 - 0 -\alpha_{b^\star}>0
\end{align*}
which contradicts our assumption that $M$ is an assignment that is popular with penalty~$\k$.

Finally, the running time of the $(\k+1)$-truncation of Algorithm~\ref{alg:popassign} is indeed as claimed, since there can be at most~$\k \cdot n$ iterations, each one taking $O(m \sqrt{n})$ time. \qed
\end{proof}

\section{Finding a popular assignment with forced/forbidden edges}
\label{sec:edge-restrictions}

In this section we consider a variant of the \myproblem{popular assignment} problem where, in addition to our instance, 
we are given a set $F^+$ of \emph{forced} edges and a set $F^-$ of \emph{forbidden} edges, 
and what we seek is a popular assignment that contains $F^+$ and is disjoint from $F^-$. 
Observe that it is sufficient to deal with forbidden edges, since putting an edge $(a,b)$ into $F^+$
is the same as putting all the edges in the set~$\{(a,b'): b' \in \Nbr(a) \textrm{ and } b' \neq b\}$ 
into $F^-$.

\paragraph{The \myproblem{popular assignment with forbidden edges} problem.}
Given a bipartite graph $G = (A \cup B, E)$
where every $a \in A$ has preferences in partial order
over her neighbors,
together with a set $F \subseteq E$ of forbidden edges,
does $G$ admit a popular assignment $M$ avoiding $F$, i.e., one where $M \cap F=\emptyset$?
\medskip

We will show that in order to deal with forbidden edges, it suffices 
to modify our algorithm in Section~\ref{sec:algo} as follows; see Algorithm~\ref{alg:popassignforbidden}.
The only difference from the earlier algorithm is that we find a maximum matching in the subgraph 
$G_{\ell} - F$, i.e., on the edge set $E_{\ell}\setminus F$.

\begin{algorithm}
\caption{Finding a popular assignment with forbidden set $F\subset E$}\label{alg:popassignforbidden}
\begin{algorithmic}[1]
\ForAll{$b \in B$} $\ell(b) = 0$. \EndFor 
\While{$\ell(b) < n$ for all $b \in B$}
\State Construct $G_\ell = (A\cup B, E_\ell)$ and find a maximum matching $M$ in $G_\ell-F$. 
\If{$M$ is a perfect matching} \textbf{return} $M$. \EndIf
\ForAll{$b \in B$ unmatched in $M$}  $\ell(b) = \ell(b) + 1$. \EndFor
\EndWhile
\State \textbf{return} ``$G$ has no popular assignment with forbidden set $F$''. 
\end{algorithmic}
\end{algorithm}

\begin{theorem}
\label{thm:edge-restrictions-algo}
The above algorithm outputs a popular assignment avoiding $F$, if such an assignment exists in $G$.
\end{theorem}
\begin{proof}
Recall that any perfect matching in $G_\ell$ is a popular assignment in $G$ by Lemma~\ref{lem:G-ell}. 
It is therefore immediate that if the above algorithm outputs a matching~$M$, 
then $M$ is a popular assignment in $G$ that avoids~$F$.

Let us now prove that if there exists a popular assignment $M^\star$ avoiding $F$, 
then our algorithm outputs such an assignment.
Let $\vec{\alpha}$ be a dual certificate for $M^\star$.
Let $\ell_i$ denote the level function at the beginning of iteration~$i$ of the algorithm.
As in the proof of Theorem~\ref{thm:pop2}, we show by induction that $\ell_i(b) \leq |\alpha_b|$ for all $b \in B$ for any iteration $i$.

This is clearly true initially with $\ell_1(b) = 0$ for all $b \in B$.
  To complete the induction, it suffices to show that $\ell_i(b_0) < |\alpha_{b_0}|$ for all $b_0 \in B$ that are unmatched by any maximum matching $M$ in~$G_{\ell_i}-F$.
  Since $M^\star$ is a perfect matching, the symmetric difference $M \oplus M^\star$ contains an $M$-augmenting path~$P$ starting at $b_0$.
  However, because $M$ has maximum size in $G_{\ell_i} - F$, the path~$P$ must contain an edge $e \notin E_{\ell_i} \setminus F$. We have $(M \cup M^\star) \cap F = \emptyset$, thus we obtain $e \notin F$ and therefore $e \notin E_{\ell_i}$. 
  Note that Lemma~\ref{lem:G-ell} implies 
  $M^\star \subseteq E_{\ell_{\vec\alpha}}$ (recall that $\ell_{\vec\alpha}(b) = |\alpha_b|$ for all $b \in B$).  
  Furthermore, $\ell_i(b) \leq \ell_{\vec\alpha}(b)$ for all $b \in B$ by our induction hypothesis.
  We can thus apply Lemma~\ref{lem:alpha-vs-ell-general} with $M' = M^\star$ and $\ell' = \ell_{\vec\alpha}$ to obtain $\ell_i(b_0) < \ell_{\vec\alpha}(b_0) = |\alpha_{b_0}|$, which completes the induction step. \qed
\end{proof}

\paragraph{\bf Minimum-cost popular assignment vs.\  popular assignment with forbidden edges.}
The \myproblem{popular assignment with forbidden edges} problem can be seen as the special case of the \myproblem{minimum-cost popular assignment} problem
in which the popular assignment may only contain edges of cost~$0$, excluding all edges of non-zero cost. In the general version of \myproblem{minimum-cost popular assignment}, however, there is a degree of freedom as to which non-zero cost edges are included in the assignment. 
Theorem~\ref{thm:mincost-nphard-01} shows that this degree of freedom introduces an additional complexity to the problem. 
Indeed, in our reduction from \myproblem{Vertex Cover} in Theorem~\ref{thm:housealloc-nphard} (which implies Theorem~\ref{thm:mincost-nphard-01}),
a set of vertices chosen as a vertex cover is encoded via the set of cost~$1$ edges chosen to be included in the 
assignment, 
and explicitly forbidding all cost~$1$ edges would turn the constructed \myproblem{popular assignment with forbidden edges} instance into a trivial `no'-instance.

\section{Finding an assignment with minimum unpopularity margin}
\label{sec:min-unpopular}

In this section we consider the \myproblem{$k$-\sf unpopularity-margin} problem in $G$;
Section~\ref{algo:k-unpop-margin} has our algorithmic result and Section~\ref{hardness:k-unpop-margin} 
contains our hardness results.

\subsection{Our algorithm}
\label{algo:k-unpop-margin}
For any assignment $M$, recall that the optimal value of \ref{LP1} is 
$\max_N \Delta(N,M)=\mu(M)$, where the maximum is taken over all assignments $N$ in $G$. 
Consequently, $\mu(M)$ equals the optimal value of the dual linear program~\ref{LP2} as well. 
Therefore, $\mu(M) = k$ if and only 
if there exists an optimal solution $\vec{\alpha}$ to \ref{LP2} 
for which $\sum_{u \in A \cup B} \alpha_u=k$. 
This leads us to a characterization of assignments with a bounded unpopularity margin that is 
a direct analog of Theorem~\ref{thm:certificate}.

\begin{theorem}
\label{thm:dual-cert-unpopular}
$M$ is an assignment with $\mu(M) \leq k$ if and only if there exists a solution $\vec{\alpha}$ to~\ref{LP2}
such that $\alpha_a \in \{0,1, \dots, n\}$ for all $a \in A$, $\alpha_b \in \{0,-1,\dots, -(n-1)\}$ for all $b\in B$,
and $\sum_{u \in A \cup B} \alpha_u \leq  k$.
\end{theorem}
\begin{proof}[of Theorem~\ref{thm:certificate} and Theorem~\ref{thm:dual-cert-unpopular}]
  If there exists 
  a feasible 
  solution~$\vec{\alpha}$ to \ref{LP2} such that $\sum_{u\in A\cup B}\alpha_u = 0$, then the
  optimal value of \ref{LP2} is 0 and hence by LP-duality, the optimal value of \ref{LP1} is also 0. Thus $\Delta(N,M) \le 0$ for any
  perfect matching $N$; in other words, $M$ is a popular assignment. Similarly, if there is 
  a feasible 
  solution~$\vec{\alpha}$ to \ref{LP2} such that $\sum_{u\in A\cup B}\alpha_u \leq k$, then $\Delta(N,M) \le k$ for any
  perfect matching~$N$, so $\mu(M) \leq k$.

  We will now show the converse. Let $M$ be a perfect matching with $\mu(M) = k$;
  there exists a dual optimal solution $\vec{\alpha}$ such that $\sum_{u\in A\cup B}\alpha_u = k$.
  Moreover, we can assume $\vec{\alpha}\in\mathbb{Z}^{2n}$ due to the total unimodularity of the constraint matrix. We can also assume $\alpha_b \leq 0$ for all $b\in B$, because feasibility and optimality are preserved if we decrease $\alpha_b$ for all $b\in B$ and increase $\alpha_a$ for all $a\in A$ by the same amount.
  
  Let us choose $\vec{\alpha}$ such that $\sum_{b\in B} \alpha_b$ is maximum subject to the above assumptions. We claim that if there is no $b\in B$ with $\alpha_b = -r$ for some $r\in \mathbb{N}$, then there is no $b \in B$ with $\alpha_b \leq -(r+1)$. Suppose the contrary, and let $B'=\{b\in B: \alpha_b<-r\}$ and $A'=\{M(b): b \in B'\}$. Since $\alpha_a+\alpha_b\geq 0$ for every $(a,b)\in M$, we have $\alpha_a\geq r+1$ for every $a \in A'$. Let $\vec{\alpha}'$ be obtained by decreasing~$\alpha_a$ by~1 for all $a\in A'$ and increasing $\alpha_b$ by 1 for all $b\in B'$. The dual feasibility constraints $\alpha'_a + \alpha'_b \ge \wt_M(a,b)$ can only be violated if $a \in A'$, 
  $b \notin B'$, and $\alpha_a + \alpha_b = \wt_M(a,b)$. But this would imply $\alpha_a\geq r+1$, $\alpha_b\geq -r+1$ (since $b \notin B'$ and $\alpha_b$ cannot be $-r$), and $\alpha_a + \alpha_b = \wt_M(a,b)\leq 1$, a contradiction. Thus, $\vec{\alpha}'$ is also an optimal dual solution, and $\sum_{b\in B} \alpha'_b> \sum_{b\in B} \alpha_b$, contradicting the choice of~$\vec{\alpha}$.
  
  We have shown that the values that $\vec{\alpha}$ takes on $B$ are consecutive integers, so we obtain that $\alpha_b \in \{0,-1,-2,$ $\ldots,-(n-1)\}$ for all $b \in B$. Since 
  $\alpha_a+\alpha_b\geq 0$ for every $(a,b)\in M$, we have $\alpha_a\geq 0$ for every $a \in A'$.
  
  To conclude the proof of Theorem~\ref{thm:certificate}, observe that $M$ is an optimal primal solution when $k=0$, so $\alpha_a+\alpha_b=0$ for every $(a,b)\in M$. This implies that $\alpha_a \in \{0,1,2, \ldots,n-1\}$ for all $a \in A$. As for  Theorem~\ref{thm:dual-cert-unpopular}, let $N$ be a perfect  matching that is optimal for~\ref{LP1}; then $\alpha_a+\alpha_b=\wt_M(a,b)\leq 1$ for every $(a,b)\in N$ by complementary slackness, and therefore $\alpha_a \leq n$ for all $a \in A$. \qed
 \end{proof}

Generalizing the notion that we already used for popular assignments, 
we define a \emph{dual certificate} for an assignment $M$ with unpopularity margin $k$
as a solution $\vec{\alpha}$ to \ref{LP2} with properties as described in Theorem~\ref{thm:dual-cert-unpopular}. 
So let us suppose that $M$ is an assignment with $\mu(M) = k$ and $\vec{\alpha}$ is a dual certificate for $M$. 
We define the \emph{load} of~$(a,b) \in M$ as $\alpha_a+\alpha_b$, 
and we will say that an edge~$(a,b) \in M$  is \emph{overloaded} (with respect to $\vec{\alpha}$), if
it has a positive load, that is, $\alpha_a+\alpha_b>0$.
Clearly, the total load of all edges in $M$ is at most $k$, moreover $\alpha_a + \alpha_b \ge \wt_M(a,b) = 0$ for every $(a,b) \in M$,  so there are at most $k$ overloaded edges in $M$.

Given a level function $\ell: B \rightarrow \mathbb{N}$ and an integer $\lambda \in \mathbb{N}$, we say that edge $(a,b)$ is 
\emph{$\lambda$-feasible}, if 
\begin{itemize}[leftmargin=24pt]
\item[(i)] $b$ has level at least $\ell^*(a)-\lambda+1$ where $\ell^*(a)=\max_{b \in \Nbr(a)} \ell(b)$, or
\item[(ii)] $b$ has level $\ell^*(a)-\lambda$ and $a$ has no neighbor in level $\ell^*(a)$ that she prefers to $b$, or
\item[(iii)] $b$ has level $\ell^*(a)-\lambda-1$, $a$ prefers $b$ to each of her neighbors in level $\ell^*(a)$, and moreover, $a$ prefers none of her neighbors in level $\ell^*(a)-1$ to $b$.
\end{itemize}
Note that $0$-feasible edges are exactly those contained in the graph $G_\ell$ induced by levels $\ell(\cdot)$, 
as defined in Section~\ref{sec:algo}. 
The following observation follows directly from the constraints of \ref{LP2}.

\begin{proposition}
\label{prop:load+feasibility}
Let $M$ be an assignment with $\mu(M) \le k$ and let $\vec{\alpha}$ be a dual certificate for $M$.
Consider the level function $\ell_{\vec{\alpha}}$ where the level of any $b \in B$ is $\ell_{\vec{\alpha}}(b)=-\alpha_b$. 
Then any edge $e \in M$ with load $\lambda$ is $\lambda$-feasible.
\end{proposition}

Given a level function $\ell: B \rightarrow \mathbb{N}$ and a 
load capacity function $\lambda: E \rightarrow \mathbb{N}$, we define the graph 
\emph{induced by levels $\ell(\cdot)$ and load capacities $\lambda(\cdot)$} as $G_{\ell,\lambda}=(A \cup B, E_{\ell,\lambda})$ 
where an edge~$e$ is contained in~$E_{\ell,\lambda}$ if and only if $e$ is $\lambda(e)$-feasible.

We are now ready to describe an algorithm for finding an assignment $M$ with $\mu(M) \leq k$ if such 
an assignment exists. Algorithm~\ref{alg:popassignmargin} starts by guessing the load $\lambda(e)$ 
for each edge $e$ of $E$. Then we use a variant of the algorithm for Theorem~\ref{thm:edge-restrictions-algo}
that enables each 
edge $e$ with $\lambda(e) > 0$ to have positive load 
(so $G_{\ell,\lambda}$ will be used instead of $G_{\ell}$), 
and  treats the overloaded edges as forced edges. 
\begin{algorithm}[h]
\caption{Finding a popular assignment with unpopularity margin at most $k$}\label{alg:popassignmargin}
\begin{algorithmic}[1]
\ForAll{functions $\lambda:E \rightarrow \mathbb{N}$ with $\sum_{e \in E} \lambda(e) \leq k$}
\State Set $K= \{e \in E: \lambda(e)>0\}$ as the edges that we will overload.
\State Set $F=\{(a,b') \in E: (a,b) \in K, b' \neq b \}$ as the set of forbidden edges.
\ForAll{$b \in  B$} $\ell(b) = 0$. \EndFor 
\While{ $\ell(b) < n$ for all $b \in B$}
\State Construct the graph $G_{\ell,\lambda}$ and find a maximum matching $M$ in $G_{\ell,\lambda} - F$. 
\If{$M$ is a perfect matching} \textbf{return} $M$. \EndIf
\ForAll{$b \in B$ unmatched in $M$} $\ell(b) = \ell(b) + 1$. \EndFor
\EndWhile
\EndFor 
\State \textbf{return} ``$G$ has no assignment $M$ with $\mu(M) \leq k$''.  
\end{algorithmic}
\end{algorithm}

Observe that there are at most $m^k$ ways to choose the load capacity function $\lambda$, where $m=|E|$, 
by the bound  $\sum_{e \in E} \lambda(e) \leq k$.
Each iteration of the while-loop takes $O(m \sqrt{n})$ time and there are at most $m^k \cdot n^2$ such iterations. 
Thus the running time of the above algorithm is $O(m^{k+1} \cdot n^{5/2})$. 

\medskip

\begin{proof}[of Theorem~\ref{thm:unpopmargin-algo-xp}]
Suppose Algorithm~\ref{alg:popassignmargin} returns an assignment~$M$. 
Consider the values of $\lambda(\cdot)$ and $\ell(\cdot)$ at the moment the algorithm outputs~$M$. 
Set $\alpha_b=-\ell(b)$ for each object, and set~$\alpha_a=-\alpha_b+\lambda(a,b)$ for each edge $(a,b) \in M$. 
From the definition of $G_{\ell,\lambda}$ and $\lambda$-feasibility, 
such a vector $\vec{\alpha}$ fulfills all constraints in \ref{LP2}.
Hence, by $\sum_{u \in A \cup B} \alpha_u = \sum_{(a,b)\in M}\lambda(a,b) \leq k$, we get that~$\mu(M) \leq k$.

Now let us assume that $G$ admits an assignment $M^\star$ with $\mu(M^\star) \leq k$, and let $\vec{\alpha}$ be a dual certificate for~$M^\star$.
We need to show that our algorithm will produce an output. 
Consider those iterations where $\lambda(e)$ equals the load of each edge $e \in M^\star$;
we call this the \emph{significant branch} of the algorithm. 
We claim that $|\alpha_b| \geq \ell(b)$ holds throughout the run of the  significant branch.

To prove our claim, we use the same approach as in the proof of Theorem~\ref{thm:edge-restrictions-algo}, based on induction. 
Clearly, the claim holds at the beginning of the branch; we need to show that $\ell(b)$ is increased only if $|\alpha_b| > \ell(b)$.
So let $|\alpha_b| \geq \ell(b)$ for each $b \in B$ at the beginning of an iteration (steps~(6)--(8)), 
and consider an object $b$ whose value is increased at the end of the iteration.

First, assume that $b$ is incident to some overloaded edge $(a,b)$ of $M^\star$ with load $\lambda(a,b)$. 
Observe that $(a,b) \in K$, and all edges incident to~$a$ other than~$(a,b)$ belong to the set~$F$ of forbidden edges.
Therefore, if $a$ is not matched to~$b$ in a maximum matching in $G_{\ell,\lambda}-F$, 
then this must be because $(a,b)$ is not an edge in $G_{\ell,\lambda}$. That is,
the edge $(a,b) \in K$ is not $\lambda$-feasible with respect to $\ell(\cdot)$.
Recall that by Proposition~\ref{prop:load+feasibility}, $(a,b)$ is $\lambda(a,b)$-feasible with respect to the level function 
$\ell_{\vec{\alpha}}$ defined by $\ell_{\vec{\alpha}}(b')=|\alpha_{b'}|$ for each object~$b'$. 
Moreover, by our induction hypothesis, $ \ell(b') \leq |\alpha_{b'}| $ holds at the beginning of the iterations for each object~$b'$, that is,
the $\ell_{\vec{\alpha}}$-level of any object is at least its $\ell$-level. 
Consider the following hypothetical procedure:
starting from the level function~$\ell(\cdot)$ (in which $(a,b)$ is \emph{not} $\lambda(a,b)$-feasible) 
we obtain the level function~$\ell_{\vec{\alpha}}(\cdot)$ (in which $(a,b)$ \emph{is} $\lambda(a,b)$-feasible)
via raising the levels of some objects and never decreasing the level of any object. 
However, by the definition of $\lambda$-feasibility, if we only raise the levels of objects other than~$b$, then 
the edge~$(a,b)$ cannot become $\lambda(a,b)$-feasible. Hence, in this hypothetical procedure of obtaining~$\ell_{\vec{\alpha}}(\cdot)$ from~$\ell(\cdot)$, 
it is necessary to raise the level of~$b$. 
It follows that $\ell_{\vec{\alpha}}(b)=|\alpha_b|> \ell(b)$  must hold.

Second, assume that $b$ is not incident to any overloaded edge. 
Consider the path $P$ in $M \oplus M^\star$ on which $b$ lies.
Notice that, since the edges of $F$ are treated as forbidden edges, 
$M$ assigns each agent incident to some edge in $K$ either the object assigned to it by $M^\star$, 
or does not assign any object to it.
Therefore, no such agent lies on the path~$P$. 
Consequently, we can apply Lemma~\ref{lem:alpha-vs-ell-general} to the matchings $M \cap P \subseteq E_{\ell}$ and $M^\star \cap P \subseteq E_{\ell_{\vec \alpha}}$. This yields $|\alpha_b| = \ell_{\vec \alpha}(b) > \ell(b)$, proving our claim.

Finally, note that in the significant branch, no object may have $\ell$-level 
higher than $n-1$, as implied by the properties of $\vec{\alpha}$ stated in Theorem~\ref{thm:dual-cert-unpopular}. \qed
\end{proof}


\subsection{Hardness results}
\label{hardness:k-unpop-margin}
We now contrast Theorem~\ref{thm:unpopmargin-algo-xp}  with Theorem~\ref{thm:unpopmargin-hardness} which states that 
finding an assignment with minimum unpopularity margin is $\mathsf{NP}$-hard
and $\mathsf{W}_l[1]$-hard with respect to the parameter $k$, our bound on the unpopularity margin. 
A parameterized problem $Q$ is $\mathsf{W}_l[1]$-hard if there exists a linear FPT-reduction 
from the \myproblem{weighted antimonotone cnf 2-sat}
(or \myproblem{wcnf 2sat$^-$}) problem\footnote{
The input to this problem is a propositional formula $\varphi$ in conjunctive normal form with only negative literals and clauses of size two, 
together with an integer parameter~$k$;
the question is whether the formula can be satisfied by a variable assignment that sets exactly~$k$ variables to true.} to~$Q$~\cite{chen-w1-vs-eth-stronger}. 
Since \myproblem{wcnf 2sat$^-$} is a $\mathsf{W}[1]$-complete problem~\cite{downey-fellows-book}, $\mathsf{W}_l[1]$-hardness implies $\mathsf{W}[1]$-hardness.
While the $\mathsf{W}[1]$-hardness of \myproblem{$k$-unpopularity margin} shows that it cannot be solved 
in~$f(k)|I|^{O(1)}$ time for any computable function~$f$ unless $W[1]=FPT$  (where $|I|$ denotes the input length), 
the results of Chen et al.~\cite{chen-w1-vs-eth-stronger} enable us to obtain a tighter lower bound: 
the $\mathsf{W}_l[1]$-hardness of \myproblem{$k$-unpopularity margin}
implies that it cannot even be solved in $f(k) |I|^{o(k)}$ time for any computable function~$f$,
unless all problems in $\mathsf{SNP}$ are solvable in subexponential time -- a possibility widely considered unlikely. 
Therefore, Theorem~\ref{thm:unpopmargin-hardness} shows that our algorithm for Theorem~\ref{thm:unpopmargin-algo-xp} is essentially optimal.

Note that the unpopularity margin of any assignment $M$ can be computed efficiently by determining the optimal value of \ref{LP1},
so the \myproblem{$k$-unpopularity margin} problem is in $\mathsf{NP}$.
In the remainder of this section, 
we present a linear FPT-reduction from the \myproblem{Clique} problem to the
\myproblem{$k$-unpopularity margin} problem where agents' preferences are weak rankings. 
By the work of Chen et al.~\cite{chen-w1-vs-eth-stronger}, the
$\mathsf{W}_l[1]$-hardness of \myproblem{$k$-unpopularity margin} follows. 
Our reduction is a polynomial-time reduction as well, implying $\mathsf{NP}$-hardness for the case of weak rankings; note that this also follows easily  from the $\mathsf{NP}$-hardness of finding a matching (not necessarily an assignment) with minimum unpopularity margin~\cite{McC08}, using our reduction from the \myproblem{popular matching} problem to the \myproblem{popular assignment} problem.
Both the reduction presented in this section and the reduction from \myproblem{popular matching} in Section~\ref{app:pop-matching-reduction} use 
weak rankings.
However, we prove the $\mathsf{NP}$-hardness of \myproblem{$k$-unpopularity margin} for strict 
rankings by reducing the problem with 
weak rankings to the case with strict rankings 
in Lemma~\ref{lem:reduction-weak-to-strict}.\footnote{The reduction from weak to strict rankings increases the parameter $k$ by a non-constant term. Thus $\mathsf{W}_l[1]$-hardness does not carry over and the parameterized complexity of \myproblem{$k$-unpopularity margin} with strict rankings is still open; see the related open question in Section~\ref{sec:open-problems}.}

Instead of giving a direct reduction from \myproblem{Clique}, we will use an intermediate problem that we call \myproblem{CliqueHog}. 
Given a graph $H$, we define a \emph{cliquehog} of size $k$ as a pair $(C,F)$ such that $C \subseteq V(H)$ is a clique of size $k$, and
$F \subseteq E(H)$ is a set of edges that contains exactly two edges connecting $c$ to $V(H) \setminus C$, for each $c \in C$.
The input of the \myproblem{CliqueHog} problem is a graph $H$ and an integer $k$, 
and it asks whether $H$ contains a \emph{cliquehog} of size $k$.

\begin{lemma}
The \myproblem{CliqueHog} problem is $\mathsf{NP}$-hard and $\mathsf{W}_l[1]$-hard with parameter~$k$.
\end{lemma}
\begin{proof}
We reduce from the \myproblem{Clique} problem. Given a graph~$H=(V,E)$ and an integer~$k$, we construct a graph~$H'$ by adding $2|V|$ edges and $2|V|$ vertices to~$H$ as follows: 
for each vertex~$v$ of~$H$ we introduce two new vertices~$v'$ and~$v''$, 
together with the edges $(v,v')$ and $(v,v'')$. 
It is then easy to see that $H$ contains a clique of size $k$ if and only if $H'$ contains a cliquehog of size $k$. 
The reduction is a linear FPT-reduction with parameter $k$, as well as a polynomial-time reduction. \qed
\end{proof}

Let us now prove the first part of Theorem~\ref{thm:unpopmargin-hardness} by presenting a reduction from the \myproblem{CliqueHog} problem. 

\paragraph{\bf Construction.}
Let $H=(V,E)$ and $k$ be our input for \myproblem{CliqueHog}. 
We construct an instance~$G$ of the \myproblem{$k$-unpopularity margin} problem, 
with a set $A$ of agents and a set~$B$ of objects as follows. 
For each vertex $v \in V$, we define a vertex gadget $G_v$ containing agents $a_v^0$ and $a_v^1$ and objects~$b_v^0$ and~$b_v^1$. 
For each $e \in E$, we define an agent $a_e$ and an object $b_e$. 
Furthermore, we will use a set~$A_D$ of~\mbox{$|E|-\binom{k}{2}-2k$} dummy agents, and a set $B_D$ of $|E|-\binom{k}{2}-2k$ artificial objects.
We define \mbox{$A_V=\{a_v^i: v \in V, i \in \{0,1\}\}$}, $B_V=\{b_v^i: v \in V, i \in \{0,1\}\}$, 
$A_E=\{A_e: e \in E\}$, and $B_E=\{b_e: e \in E\}$. We set $A=A_V \cup A_E \cup A_D$ and $B=B_V \cup B_E \cup B_D$.
The preferences of the agents in $G$ are as follows (ties are simply indicated by including them as a set in the preference list):

\medskip 
\begin{tabular}{l@{\hspace{2pt}}ll}
$a_v^i:$ & $\{b_e : e \textrm{ is incident to $v$ in $H$} \} \succ b_v^0 \succ b_v^1 \qquad $ & for each $v \in V$ and $i \in \{0,1\}$; \\
$a_e:$ & $b_e \succ B_D \cup \{b_x^0,b_y^0\} \succ \{b_x^1,b_y^1\}$  & for each $e=(x,y) \in E$; \\
$a_d:$ & $B_E$ & for each $a_d \in A_D$. \\
\end{tabular}
\medskip

We finish the construction by setting the bound for the unpopularity margin of the desired assignment as $k$. 
Clearly, this is a polynomial-time reduction, and also a linear FPT-reduction with parameter $k$, so it remains to prove 
that $H$ contains a cliquehog of size $k$ if and only if $G$ admits an assignment $M$ with unpopularity margin at most $k$. 

\begin{lemma}
\label{lem:unpopmargin-reduction-dir1}
If $(C,F)$ is a cliquehog in  $H$ of size $k$, then $G$ admits an assignment~$M$ with unpopularity margin at most $k$. 
\end{lemma}

\begin{proof}
Let $f_c^0$ and $f_c^1$ denote the two edges of $F$ connecting $c$ to $V \setminus C$ in $H$ (in any fixed order), 
and we set $F^i=\{f_c^i: c \in C\}$ for $i \in \{0,1\}$. 

\begin{figure}[h]
\begin{center}
\includegraphics[width=\linewidth]{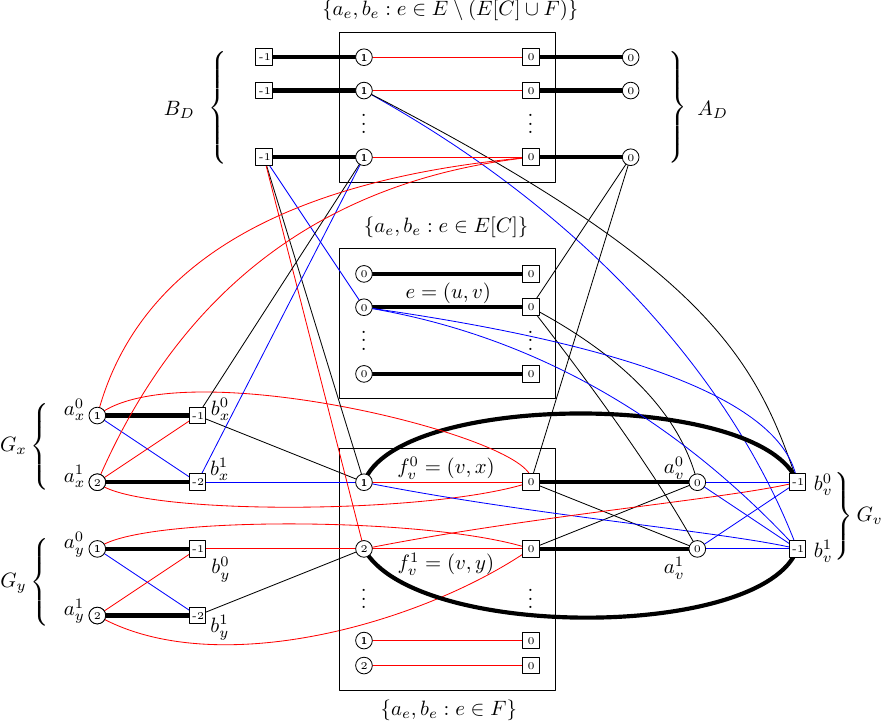}
\caption{Illustration for the construction in the proof of Theorem~\ref{thm:unpopmargin-hardness}.
The assignment $M$ defined in Lemma~\ref{lem:unpopmargin-reduction-dir1}
is indicated by bold lines.
Red, black, and blue edges have weight $+1$, $0$, and~$-1$, respectively, according to $\wt_M(\cdot)$. 
The values of the dual certificate $\vec{\alpha}$ for $M$ are indicated by numbers within the circle (square) 
corresponding to the given agent (object, respectively). 
For some edges $(a_e,b_e)$ in $G$ we indicate only the corresponding edge $e$ of $H$ (see $(u,v)$, $(v,x)$ and $(v,y)$).
The figure assumes $v \in C$ but $x,y \notin C$; 
note that the edge $(a_{(v,y)},b_v^1)$ is overloaded in~$M$.
}
\label{fig:margin-gadgets}
\end{center}
\end{figure}

We define an assignment $M$ as follows; see Fig.~\ref{fig:margin-gadgets} as an illustration.
First, let us assign the $|E|-\binom{k}{2}-2k$ objects in $\{b_e: e \notin E[C] \cup F\}$ to the dummy agents 
(where $E[C]$ denotes the set of those edges of $E$ whose both endpoints are in $C$).
Second, we assign the artificial objects to the $|E|-\binom{k}{2}-2k$ agents in $\{a_e: e \notin E[C] \cup F\}$.
To define $M$ on the remaining objects and agents, let

\medskip
\begin{tabular}{ll}
$M(a_v^i)=b_v^i$ & for each $v \in V \setminus C$ and $i \in \{0,1\}$; \\
$M(a_v^i)=b_{f_v^i}$ \qquad \qquad & for each $v \in C$ and $i \in \{0,1\}$; \\
$M(a_e)=b_e$ & for each $e \in E[C]$; \\
$M(a_f)=b_v^i$ & for each $f \in F$ where $f=f_v^i$.
\end{tabular}
\medskip

Observe that $M$ indeed assigns exactly one object to each agent. 
To show that $\mu(M) \leq k$, we define a dual certificate $\vec{\alpha}$  for $M$:

\medskip
\begin{tabular}{ll}
$\alpha_{a_d}=0$ & for each $a_d \in A_D$;  \\
 $\alpha_{b_d}=-1$ & for each $b_d \in B_D$;  \\
$\alpha_{a_v^0}=1$ \qquad \qquad & for each $v \in V \setminus C$; \\
 $\alpha_{b_v^0}=-1$ \qquad \qquad & for each $v \in V \setminus C$; \\
 $\alpha_{a_v^1}=2$ & for each $v \in V \setminus C$; \\
 $\alpha_{b_v^1}=-2$ & for each $v \in V \setminus C$; \\
 $\alpha_{a_v^i}=0$ & for each $v \in C$ and $i \in \{0,1\}$; \\
 $\alpha_{b_v^i}=-1$ & for each $v \in C$ and $i \in \{0,1\}$; \\
$\alpha_{a_e}=1$ & for each $e \in E \setminus (F^1 \cup E[C])$; \\
 $\alpha_{b_e}=0$ & for each $e \in E$;  \\
$\alpha_{a_e}=2$ & for each $e \in F^1$;  \\
$\alpha_{a_e}=0$ & for each $e \in E[C]$. 
\end{tabular}
\medskip

It is straightforward to check that $\vec{\alpha}$ is indeed a dual certificate, that is, it satisfies the constraints of \ref{LP2}.
Since $\sum_{u \in A \cup B} \alpha_u=k$, assignment $M$ indeed has unpopularity margin at most $k$. \qed
\end{proof}

We now show that any assignment in $G$ with unpopularity margin at most $k$ implies the existence of a cliquehog of size $k$ in $H$. We first establish a useful assumption that we will show is without loss of generality.
We say that an assignment $M$ has a \emph{nice structure}, if for each $e \in E$ one of the following cases holds true:
\begin{itemize}
  \item[$\bullet$] $M(a_e) = b_e$, or
  \item[$\bullet$] $M(a_e) = b^i_x$ and $M(a^i_x) = b_e$ for some endpoint $x$ of edge $e$ in $H$ and $i \in \{0, 1\}$, or
  \item[$\bullet$] $M(a_e) \in B_D$ and $M(b_e) \in A_D$.
\end{itemize}

\begin{lemma}
If there exists an assignment in $G$ with unpopularity margin at most $k$, then there exists an assignment in $G$ with unpopularity margin at most $k$ that has a nice structure.
\end{lemma}

\begin{proof}
  Let $M'$ be any assignment in $G$.
  By switching the names of the identical agents~$a^0_x$ and~$a^1_x$ wherever necessary, we can assume that $M'(a^i_x) \in \{b^i_x\} \cup B_E$ for all $x \in V$ and $i \in \{0, 1\}$.
   We construct a new assignment~$M$ as follows: For each $e \in E$, let $M(a_e) = M'(a_e)$. For each $x \in V$ and each $i \in \{0, 1\}$, if $M'(b^i_x) = a_e$ for some $e \in E$, then let $M(a^i_x) = b_e$; otherwise let $M(a^i_x) = b^i_x$. Assign the unmatched objects in $B_E$ arbitrarily to the dummy agents in $A_D$. Note that for every agent~$a \in A$, we have $M(a) = M'(a)$ or $M(a)$ and $M'(a)$ are both in $B_E$. 
   Since every agent is indifferent between all her neighbors in $B_E$, the assignment $M$ has the same unpopularity margin as $M'$. \qed
\end{proof}

In what follows, let $M$ be an assignment in $G$ with unpopularity margin at most $k$ and a nice structure, and let $\vec{\alpha}$ be a dual certificate for $M$. We will construct a cliquehog of size $k$ in $H$. 

We define a partition $(E_0,E_1,E_2)$ of those edges $e$ in $E$ for which $M(a_e) \notin B_D$. 
For any edge~$e \in E$, we decide which set of the partition $e$ belongs to using the following procedure.
\begin{enumerate}
\item If $M(a_e) = b_e$ and $\alpha_{a_e}+\alpha_{b_e} > 0$, then we put $e$ into $E_0$.
\item If $M(a_e) = b_e$ and $\alpha_{a_e}+\alpha_{b_e}=0$, then we put $e$ into $E_1$. 
\item If $M(a_e)$ is an object in $G_x$ or $G_y$ where $e=(x,y)$, then we put $e$ into $E_2$.
\end{enumerate}
We also define $S$ as the set of those vertices $v \in V$ for which the vertex gadget $G_v$ contains 
an agent or object that is incident to an overloaded edge of $M$. 
Note that since no agent in a vertex gadget is connected to an object in another vertex gadget, 
and moreover, the overloaded edges of~$M$ corresponding to edges in~$E_0$ are not incident to any vertex gadget, 
we know that 
\begin{equation}
\label{eq:overloaded-edges}
|S|+ |E_0|\leq \mu(M) \leq k.
\end{equation} 

We will show  that $(S, E_2)$ is a cliquehog of size $k$. To do so, we first establish two helpful lemmas. The first shows that $S$ contains all nodes $x \in V$ for which at least one of the objects $b^0_x$ and $b^1_x$ is matched to an edge agent. This implies that every edge in $E_2$ is incident to a node in $S$.
The second lemma shows that all endpoints of edges in $E_1$ are contained in $S$.

\begin{lemma}
\label{lem:E2}
Let $x \in V$. If $M(b^0_x) \in A_E$ or $M(b^1_x) \in A_E$, then $x \in S$.
\end{lemma}

\begin{proof}
We distinguish two cases. First, suppose that $M(b_x^0) = a_e$ for some $e \in E$. Recall that this implies $M(a^0_x) = b_e$, because $M$ has a nice structure.
From the preferences of agents and the feasibility of~$\vec{\alpha}$, we get the three inequalities $\alpha_{a^0_x} + \alpha_{b^1_x} \geq -1$ and
$\alpha_{M(b^1_x)} + \alpha_{b^0_x} \geq 1$ and
$\alpha_{a_e} + \alpha_{b_e} \geq 1$. Adding these inequalities, we observe that at least one of the 
sums $\alpha_{a^0_x} + \alpha_{b_e}$ or~$\alpha_{b^1_x} + \alpha_{M(b^1_x)}$ or~$\alpha_{a_e} + \alpha_{b^0_x}$ must be positive, and thus the corresponding edge in $M$ must be overloaded.

Second, suppose that $M(b^0_x) \notin A_E$ but $M(b^1_x) = a_e$ for some $e \in E$. 
Recall that this implies  $M(a^0_x) = b^0_x$ and $M(a^1_x) = b_e$, because $M$ has a nice structure.
Again from the preferences of  agents and the feasibility of $\vec{\alpha}$, we obtain the three inequalities $\alpha_{a_e} + \alpha_{b_x^0} \geq 1$ and $\alpha_{a_x^0} + \alpha_{b_e} \geq 1$ and $\alpha_{a_x^1} + \alpha_{b_x^1} \geq -1$.
Adding these inequalities, we see that at least one of the 
sums $\alpha_{a_e} + \alpha_{b_x^1}$ or $\alpha_{a_x^0} + \alpha_{b_x^0}$ or $\alpha_{a_x^1} + \alpha_{b_e}$ must be positive, and so the corresponding edge in $M$ must be overloaded. \qed
\end{proof}

\begin{lemma}
\label{cor:E1}
If $e = (x,y) \in E_1$, then $\{x,y\} \subseteq S$.
\end{lemma}
\begin{proof}
We only show that $x \in S$. The proof for $y \in S$ follows by symmetry.
If $M(b^0_x) \in A_E$ or~$M(b^1_x) \in A_E$, Lemma~\ref{lem:E2} implies that $x \in S$. So we can assume that $M(b^0_x) = a^0_x$ and $M(b^1_x) = a^1_x$ by the nice structure of $M$.
By the preferences of agents and feasibility of $\vec{\alpha}$, we obtain the three inequalities
$\alpha_{a^0_x} + \alpha_{b_e} \geq 1$ and $\alpha_{a^1_x} + \alpha_{b^0_x} \geq 1$ and $\alpha_{a_e} + \alpha_{b^1_x} \geq -1$. Adding up these inequalities, we observe that at least one of the 
sums $\alpha_{a^0_x} + \alpha_{b^0_x}$ or $\alpha_{a^1_x} + \alpha_{b^1_x}$ or $\alpha_{a_e} + \alpha_{b_e}$ must be positive. By~$e \in E_1$, the third expression is equal to $0$, and so one of the former two has to be positive. This implies that at least one of the two corresponding edges of $M$ in $G_x$ is overloaded. \qed
\end{proof}

We are now ready to prove Lemma~\ref{lem:unpopmargin-reduction-dir2}, which
together with Lemma~\ref{lem:unpopmargin-reduction-dir1} proves Theorem~\ref{thm:unpopmargin-hardness}.

\begin{lemma}
\label{lem:unpopmargin-reduction-dir2}
If the constructed instance~$G$ admits an assignment $M$ with unpopularity margin at most $k$, then 
$H$ contains a cliquehog of size $k$.
\end{lemma}

\begin{proof}
By Lemma~\ref{cor:E1} we know that any edge $e \in E_1$ must have both of its endpoints in $S$,
yielding 
\begin{equation}
\label{eq:E1+E5}
|E_1| \leq \binom{|S|}{2}. 
\end{equation}
By Lemma~\ref{lem:E2}, each edge $e \in E_2$ must have its agent~$a_e$ assigned to an object in a vertex gadget~$G_v$ for some~$v \in S$. By~Inequality~(\ref{eq:overloaded-edges}) there are at most $2|S|$ such objects, so we obtain
\begin{equation}
\label{eq:E2}
|E_2| \leq 2|S|.
\end{equation}

Recall that by construction of $G$ and the definition of the partition $(E_0,E_1,E_2)$, we know 
\begin{align*}
\binom{k}{2}+2k = |E|-|B_D|= |E_1| + |E_2| + |E_0| \leq \binom{|S|}{2} + |S| + k
\end{align*}
where the inequality follows from combining Inequalities~(\ref{eq:overloaded-edges}), (\ref{eq:E1+E5}), and (\ref{eq:E2}).
Hence, by $|S| \leq k$ we obtain $|S|=k$. Moreover, every inequality we used must hold with equality. 
In particular,  this implies $|E_1|=\binom{k}{2}$, which can only happen if there are $\binom{k}{2}$ edges in $H$ 
(namely, those in $E_1$) with both of their endpoints in $S$.
In other words, $S$ forms a clique in $H$. 
Additionally, (\ref{eq:E2}) must also hold with equality, so $|E_2|=2|S|=2k$. 
Since every edge in $E_2$ is incident to a vertex of $S$ by Lemma~\ref{lem:E2}, but 
is not contained in $E[S]$ 
(because $E[S]=E_1 \subseteq E \setminus E_2$), and any $x \in S$ 
is incident to at most two edges of $E_2$ 
(by the definition of $E_2$), we can conclude that $(S,E_2)$ is a  cliquehog of size~$k$. \qed
\end{proof}

In order to prove the second part of Theorem~\ref{thm:unpopmargin-hardness}, we show in Lemma~\ref{lem:reduction-weak-to-strict} a 
reduction from weak rankings to strict for the \texorpdfstring{$k$}{k}-\myproblem{unpopularity margin} problem.
The proof of Lemma~\ref{lem:reduction-weak-to-strict} is given in the appendix.

\begin{restatable}{lemma}{lemweaktostrict}
\label{lem:reduction-weak-to-strict}
Let $G$ be an instance with weak rankings and $n$ agents. Then we can compute in polynomial time an instance $G'$ with strict rankings and an integer $q \in \mathcal{O}(n)$ such that $G$ admits an assignment with unpopularity margin $k$ if and only if $G'$ admits an assignment with unpopularity margin $k+q$ for any $k \in [n]$. 
\end{restatable}
This completes the proof of Theorem~\ref{thm:unpopmargin-hardness}.

\section{Popular allocations in housing markets}
\label{sec:housing}
A housing market, as introduced by Shapley and Scarf in their seminal paper~\cite{SS74}, 
models a situation where agents are initially endowed 
with one unit of an indivisible good, called a house, and agents may trade
their houses according to their preferences without using monetary transfers. 
In such markets, trading results in a reallocation of houses in a way that each agent
ends up with exactly one house. 

Let $A$ denote the set of agents, and let $\omega(a)$ denote the house of agent~$a \in A$;
note that agents do not share houses, so $\omega(a) \neq \omega(a')$ whenever $a \neq a'$. 
The classic representation of a \emph{housing market} over $A$ is
a directed graph~$D=(A,F)$ where each agent~$a \in A$ has preferences over the arcs leaving~$a$ in $D$. 
In this model, an arc $(a,a') \in F$ means that $a$ prefers $\omega(a')$ to her own house (i.e.,~$\omega(a)$), 
and is therefore willing to participate in a trade where she obtains $\omega(a')$. 
Thus, preferences over the arcs leaving~$a$ describe her preferences of the houses she finds acceptable (i.e., prefers to her own).
An \emph{allocation} in a housing market~$D=(A,F)$ is a set $S \subseteq F$ of edges that form disjoint 
cycles in~$D$ called \emph{trading cycles}. An agent \emph{trades} in~$S$ if it is covered by a trading cycle of~$S$.

The solution concept used by Shapley and Scarf is the so-called \emph{core}~\cite{SS74}, but there have been several other stability notions studied in the literature on housing markets. 
We may consider popularity as an alternative concept, with a  straightforward interpretation as follows.
Given allocations~$S$ and~$S'$ in $D$, some agent~$a \in A$ \emph{prefers}~$S$ to~$S'$, 
if either $S$ covers $a$ but $S'$ does not,
or $(a,a_1) \in S$, $(a,a_2) \in S'$ and $a$ prefers~$\omega(a_1)$ to~$\omega(a_2)$. 
We say that $S$ is \emph{more popular than~$S'$}, if the number of agents preferring~$S$ to~$S'$ is more than the number of agents
preferring~$S'$ to~$S$. An allocation~$S$ is \emph{popular} in~$D$, if there is no allocation in~$D$ that is more popular than~$S$. 
A natural question is the following: can we find a popular allocation in a given housing market  efficiently?

\subsection{Popular allocations in a housing market: reduction to \myproblem{popular assignment}}
We propose a reduction from the problem of finding a popular allocation in a housing market 
to the problem of finding a popular assignment in a bipartite graph. 
The main idea of the reduction is to associate a bipartite graph $G_D$ 
with a given housing market~$D=(A,F)$ as follows.
We define the graph 
$G_D = (A \cup B, E)$ where objects are houses, i.e., $B=\{\omega(a) : a \in A\}$,
and the edge set is $E=\{(a,\omega(a')):(a,a') \in F)\} \cup \{(a,\omega(a)):a \in A\}$.
Thus, each arc $(a,a') \in F$ corresponds to an edge $(a, \omega(a'))$ in~$E$, and we also add
an edge~$(a,\omega(a))$ for each agent~$a \in A$, corresponding to the possibility of~$a$ staying in her own house.
Preferences of an agent~$a$ over her outgoing arcs in~$D$ naturally yield her preferences in $G_D$ over objects
in~$\Nbr_{G_D}(a) \setminus \{\omega(a)\}$, while $\omega(a)$ is the unique worst-choice object of~$a$.

Observe that there is a one-to-one correspondence between allocations in~$D$ and assignments in~$G_D$
that preserves the preferences of any agent: an allocation~$S$ in $D$
corresponds to an assignment $M_S=\{(a,\omega(a')):(a,a') \in S\} \cup \{(a,\omega(a)): a$ does not trade in $S\}$;
note that this indeed yields a bijection.
It is easy to see that an allocation~$S$ is popular in~$D$ if and only if $M_S$ is a popular assignment in~$G_D$.
Hence, we can find a popular allocation in a housing market~$D$ by using Algorithm~\ref{alg:popassign} 
to find a popular assignment in~$G_D$; thus the following is a corollary to Theorem~\ref{thm:algo}.

\begin{corollary}
\label{prop:pop_allocation}
For any housing market $D=(A,F)$ 
on $n$ nodes and $m$ arcs where each agent has preferences over her outgoing arcs, 
we can decide in time~$O((m+n) \cdot n^{5/2})$ whether $D$ admits a popular allocation or not, and if so, find one.
\end{corollary}

\subsection{Maximum-size popular allocation in a housing market}
Given the above positive result (Corollary~\ref{prop:pop_allocation}), an interesting question is whether it is 
possible to find a maximum-size popular allocation in a housing market $D=(A,F)$, 
where the \emph{size} of an allocation~$S$ in~$D$ is~$|S|$, that is, the number of agents trading in~$S$.
A natural approach for finding a maximum-size popular allocation in~$D$ would be to assign 
cost~$1$ to the edge~$(a,\omega(a))$ for each $a \in A$ in~$G_D$, assign cost~$0$ to all remaining edges of~$G_D$,
and ask for a minimum-cost popular assignment in $G_D$.
However, as we prove in Theorem~\ref{thm:housealloc-nphard}, finding a popular allocation of maximum size in a housing market is $\mathsf{NP}$-hard;
by the above reasoning, this implies that finding a minimum-cost popular assignment in the classic setting of the 
\myproblem{popular assignment} problem is also $\mathsf{NP}$-hard, even for binary costs. 
Thus, Theorem~\ref{thm:mincost-nphard-01} is a corollary of Theorem~\ref{thm:housealloc-nphard}.

\begin{theorem}
\label{thm:housealloc-nphard}
The \myproblem{minimum-cost popular assignment} problem is $\mathsf{NP}$-hard, even if agents' preferences are strict,
the input graph~$G_D$ corresponds to a housing market~$D=(A,F)$ where edges in~$\{(a,\omega(a)) : a \in A\}$, corresponding to initial endowments in~$D$, have cost~1, while 
all remaining edges have cost~0. 
\end{theorem}
\begin{proof}
We present a reduction from the \myproblem{Vertex Cover} problem, whose input is a graph $H=(V,E)$ and an integer $k$, 
and asks whether there exists a set $S$ of at most $k$ vertices in $H$ such that each edge of $E$ has one of its endpoints in $S$. 
We will construct an equivalent instance of the \myproblem{minimum-cost popular assignment} problem.
We assume an arbitrary, fixed ordering on the vertices of $H$, which allows us to represent each edge in $E$ as an ordered pair $(x,y) \in E$; this implicit ordering will be relevant in our construction.

\paragraph{\bf Gadgets.}
We begin by defining two types of gadgets.  
A \emph{gadget} $G_\gamma$ of size $s \in \mathbb{N}$ 
contains agents~$a_{\gamma}^i$ and objects~$b_{\gamma}^i$ for each $i \in \{0,1,\dots, 6s-1\}$.
Superscripts of these agents and objects are always taken modulo $6s$, 
in particular, we set $a_\gamma^{6s}=a_\gamma^0$ and $b_\gamma^{6s}=b_\gamma^0$. 
We let $A_\gamma=\{a_\gamma^i:i \in [6s]\}$ and $B_\gamma=\{b_\gamma^i:i \in [6s]\}$.

Preferences of agents in $A_\gamma$ depend on the type of the gadget;  see Fig.~\ref{fig:HA-gadget}.
If $G_\gamma$ is a \emph{type-0 gadget}, then agents in~$A_\gamma$ have preferences as follows.

\medskip
\begin{tabular}{l@{\hspace{2pt}}ll}
$a_{\gamma}^i:$ & $b_\gamma^i \succ b_\gamma^{i+1}$ & for each $i \in \{0,1, \dots, 3s-1\}$; \\
$a_{\gamma}^i:$ & $b_\gamma^{i+1} \succ b_\gamma^i$ & for each $i \in [6s-1] \setminus [3s-1]$ where $i \not \equiv 0 \mod 3$; \\
$a_{\gamma}^i:$ & $b_\gamma^{i+1} \succ b_\gamma^i \succ b_\gamma^{i+3} $ &  for each $i \in [6s-1] \setminus [3s-1]$ where $i \equiv 0 \mod 3$. 
\end{tabular}
\medskip

\begin{figure}[t]
\begin{center}	
\includegraphics[scale=1]{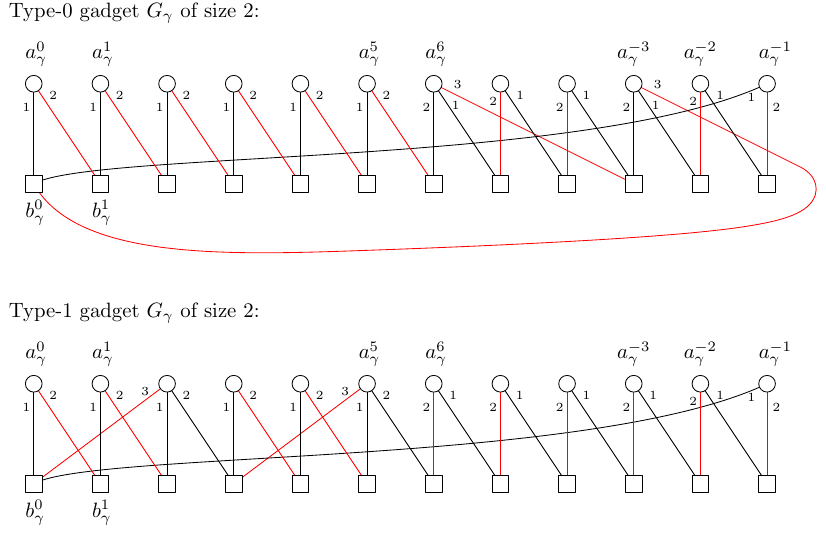}
\caption{An illustration of type-0 and type-1 gadgets constructed in the proof of Theorem~\ref{thm:housealloc-nphard}.
Edges leading from an agent to its worst-choice object are shown in red.
}
\label{fig:HA-gadget}
\end{center}
\end{figure}

By contrast, if $G_\gamma$ is a \emph{type-1 gadget}, then its agents have the following preferences.

\medskip
\begin{tabular}{l@{\hspace{2pt}}ll}
$a_{\gamma}^i:$ & $b_\gamma^i \succ b_\gamma^{i+1}$ & for each $i \in \{0,1, \dots, 3s-1\}$  where $i \not \equiv 2 \mod 3$; \\
$a_{\gamma}^i:$ & $b_\gamma^i \succ b_\gamma^{i+1} \succ b_\gamma^{i-2} $ &  for each $i \in  \{0,1, \dots, 3s-1\}$ where $i \equiv 2 \mod 3$; \\
$a_{\gamma}^i:$ & $b_\gamma^{i+1} \succ b_\gamma^i$ & for each $i \in [6s-1] \setminus [3s-1]$.
\end{tabular}
\medskip

Furthermore, we define two disjoint matchings in gadget $G_\gamma$, namely $M_\gamma^0=\{(a_\gamma^i,b_\gamma^i):i \in [6s]\}$
and $M_\gamma^1=\{(a_\gamma^i,b_\gamma^{i+1}):i \in [6s]\}$. 

Notice that the edge set~$\{(a,\omega(a)):a \in A_\gamma \}$ leading from the agents to their worst-choice objects
forms a perfect matching within the gadget.
In our construction, we will assign a cost of~1 for such edges, setting all remaining edge costs as~0. 
This gives rise to the following important difference between type-0 and type-1 gadgets:
for a type-0 gadget~$G_\gamma$ of size~$s$, 
	the cost of~$M_\gamma^0$ is~$2s$, while the cost of~$M_\gamma^1$ is~$3s$.
By contrast, for a type-1 gadget~$G_\gamma$ of size~$s$, 
	the matching~$M_\gamma^0$ has cost~$3s$, while $M_\gamma^1$ has cost~$2s$. 
In fact, the sole purpose of the differences between the two types of gadgets 
is to make $M_\gamma^0$ ``cheap'' and $M_\gamma^1$ ``expensive'' in a type-0 gadget, and the other way around in a type-1 gadget.

\paragraph{\bf Construction.}
Let us set $n_v=3|E|+1$ and $n_\ell=(3|V|+1)n_v$; it will be sufficient for the reader to keep $n_\ell \gg n_v \gg |E|$ in mind.
We are now ready to define the set~$A$ of agents and the set~$B$ of objects by creating the following: 
\begin{itemize}
\item a type-0 gadget~$G_\ell$ of size~$n_\ell$ called the \emph{level-setting gadget}, 
\item for each $e \in E$ a type-0 gadget~$G_e$ of size $1$ called an \emph{edge gadget}, and
\item for each $x \in V$ a type-1 gadget~$G_x$ of size $n_v$ called a \emph{vertex gadget}.
\end{itemize}
Additionally, we introduce a set $F$ of \emph{inter-gadget edges}; see Fig.~\ref{fig:HA-consruction} for an illustration. 
First, in order to enforce certain lower bounds on the dual certificate, 
we connect some agents and objects in the level-setting gadget with those in the vertex and edge gadgets, by adding the edges of
\begin{equation}
\label{eqn-HA-Fbnd-def}
 F_{\textup{bnd}}=\{(a_x^{-3},b_\ell^{-2}):x \in V \} \cup \{(a_e^2,b_\ell^{-3}),(a_e^4,b_\ell^{-2}),(a_\ell^{-3},b_e^4): e \in E\}.
\end{equation}
Second, we encode the incidence relation in $H$ into our instance by adding the edges of
\begin{equation}
\label{eqn-HA-Finc-def}
F_{\textup{inc}} = \{ (a_e^2,b_x^{-3}), (a_e^3,b_y^{-3}) : e= (x,y) \in E\}
\end{equation}
between edge and vertex gadgets.
We define the set of all inter-gadget edges as $F= F_{\textup{bnd}} \cup F_{\textup{inc}}$.
Note that the edges of $F$ indeed run between different gadgets.

Preferences of agents incident to inter-gadget edges are given below; all remaining agents have preferences as given in the 
definition of a type-0 or type-1 gadget. 
Again, for a set $X$ of objects, we  write $[X]$ in an agent's preference list to denote an arbitrarily ordered list containing objects in~$X$.

\medskip
\begin{tabular}{lll}
$a_{\ell}^{-3}:$ & $b_\ell^{-2} \succ b_\ell^{-3} \succ [\{b_e^4:e \in E\}] \succ b_\ell^0$; $\qquad \quad$ & \\
$a_e^2:$ & $b_\ell^{-3} \succ b_x^{-3} \succ b_e^2 \succ b_e^3$ & for each $e=(x,y) \in E$; \\
$a_e^3:$ & $ b_y^{-3} \succ b_e^4 \succ b_e^3 \succ b_e^0$ & for each $e=(x,y) \in E$; \\
$a_e^4:$ & $b_\ell^{-2} \succ b_e^5 \succ b_e^4$  & for each $e=(x,y) \in E$; \\
$a_x^{-3}:$ & $b_\ell^{-2} \succ b_x^{-2} \succ b_x^{-3}$ & for each $x \in V$.
\end{tabular}
\medskip

Note that the introduction of inter-gadget edges did not change the worst-choice object of any agent. 
As a consequence, the edge set $\{(a,\omega(a)):a \in A\}$ forms a perfect matching in the resulting graph~$G$, and
therefore $G$ indeed corresponds to a housing market as promised. 

\begin{figure}[t]
\begin{center}	
\includegraphics[width=\linewidth]{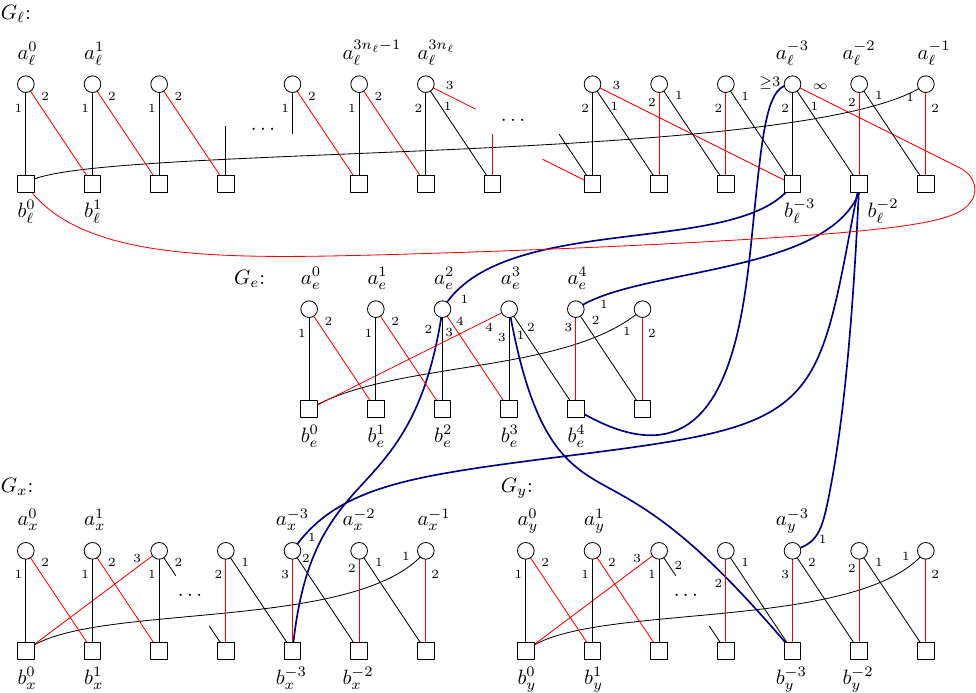}
\caption{An illustration of the construction in the proof of Theorem~\ref{thm:housealloc-nphard}.
The figure assumes $e=(x,y) \in E$. 
Again, worst-choice edges are shown in red, while inter-gadget edges are depicted in blue.
}
\label{fig:HA-consruction}
\end{center}
\end{figure}

We set the cost of all edges in $\{(a,\omega(a)):a \in A\}$ as~$1$.
All remaining edges have cost~$0$, and we set $\beta=2n_\ell+3kn_v+(|V|-k)2n_v+3|E|$ as our budget. 
This finishes the definition of our instance.

\paragraph{\bf Direction ``$\Rightarrow$''.}
To prove the correctness of our reduction, 
first we will show that the existence of a popular assignment $M$ of cost at most $\beta$ in $G$ 
implies a vertex cover of size at most $k$ in $H$. 
Let us make some basic observations about the structure of~$M$. 

\paragraph{\bf Excluding inter-gadget edges.}
We claim that $M$ does not contain any inter-gadget edges. 
To prove this, we first show that $(a_\ell^{-2},b_\ell^{-2}) \in M$.
For any index $i \in \{-2,-1, 0, \dots, 3n_\ell-1\}$, agent~$a_\ell^i$ can be assigned either~$b_\ell^i$ or~$b_\ell^{i+1}$. 
Using that $M$ assigns an object to every agent, we immediately 
get that  if $(a_\ell^{-2},b_\ell^{-2}) \notin M$, then $M$ must contain all 
edges $(a_\ell^i, b_\ell^{i+1})$ for $i \in \{-2,-1,0,1, \dots, 3n_\ell-1\}$. 
However, these edges have total cost $3n_\ell$ which by our choice of $n_\ell$ exceeds the budget~$\beta$. 
Hence, $M$ contains $(a_\ell^{-2},b_\ell^{-2})$. 
In particular, this shows that agents in a vertex gadget~$G_x$ can only be assigned by $M$ an object within the same gadget.
Thus, all objects within vertex gadgets must be assigned to agents within vertex gadgets,
implying that no edge of~$F$ incident to a vertex gadget may be contained in $M$. 

To prove our claim, it remains to show that $M$ does not contain any edge connecting the level-setting gadget with an edge gadget. 
Assume otherwise for the sake of contradiction.
Recall that only $a_\ell^{-3}$ and $b_\ell^{-3}$ may be incident to an edge of $M \cap F$ among all agents and objects  in~$G_\ell$. 
Thus, it follows that $M$ must contain both $(a_\ell^{-3},b_e^4)$ and $(a_e^2,b_\ell^{-3})$ for some $e \in E$.
(Observe that here we used the fact that no edge leaving a vertex gadget may be in $M$.)
From this we also get $(a_e^3,b_e^3) \in M$. We can then define a matching $M'$ by replacing these three edges with the 
edges $(a_\ell^{-3},b_\ell^{-3})$, $(a_e^2,b_e^3)$ and $(a_e^3,b_e^4)$. Observe that $M'$ is more popular than $M$, 
since both $a_e^3$ and $a_\ell^{-3}$ prefer $M'$ to $M$; a contradiction to the popularity of $M$.
Hence, we have established that $M \cap F=\emptyset$. 

\paragraph{\bf Assignment~$M$ on the level-setting gadget.}
We show that $M_\ell^0= \{(a_\ell^i,b_\ell^i):i \in [6n_\ell]\} \subseteq M$.
To this end, we first prove that $M$ does not contain any of the edges~$(a_\ell^i,b_\ell^{i+3})$ where $i \in [6n_\ell]$ 
and~$i \equiv 0 \mod 3$.
Suppose $(a_\ell^i,b_\ell^{i+3}) \in M$ for the sake of contradiction. 
Then $(a_\ell^{i+2},b_\ell^{i+2}) \in M$ and~$(a_\ell^{i+1},b_\ell^{i+1}) \in M$ follows. 
However, we can now define an assignment~$M'$ by deleting these three edges from $M$ and adding the edges
$(a_\ell^{i},b_\ell^{i+1})$, $(a_\ell^{i+1},b_\ell^{i+2})$, and~$(a_\ell^{i+2},b_\ell^{i+3})$ instead. 
Notice that all three agents $a_\ell^i$,  $a_\ell^{i+1}$, and $a_\ell^{i+2}$ prefer~$M'$ over~$M$, contradicting our assumption that $M$ is popular.

Together with $M \cap F=\emptyset$, this implies that $M(a_\ell^i) \in \{ b_\ell^i,b_\ell^{i+1}\}$ for each $i \in [6n_\ell]$. 
Recall that we already know $(a_\ell^{-2},b_\ell^{-2}) \in M$, which yields $M_\ell^0 \subseteq M$.

\paragraph{\bf Edge and vertex gadget: a binary choice for $M$.}
Consider some edge gadget $G_e$.
Since $M$ does not contain any inter-gadget edges, we can argue that $M$ does not contain the edge~$(a_e^3,b_e^0)$:
indeed, 
supposing $(a_e^3,b_e^0) \in M$ would imply $\{(a_e^4,b_e^4),(a_e^5,b_e^5)\} \subseteq M$; however, 
replacing these edges by $\{(a_e^i,b_e^{i+1}):i =3,4,5\}$
yields a matching more popular than~$M$, a contradiction. 
Hence, $M$ contains either $M_e^0$ or $M_e^1$.

Analogously, we can show that $M$ contains no edge of the form~$(a_x^i,a_x^{i-2})$ in a vertex gadget~$G_x$:
indeed, $(a_x^i,b_x^{i-2}) \in M$ would imply $\{(a_x^{i-2},b_x^{i-1}),(a_x^{i-1},b_x^i)\} \subseteq M$; however, 
replacing these edges by $\{(a_e^j,b_e^j):j =i-2,i-1,i\}$ yields a matching more popular than~$M$, a contradiction. 
Hence, $M$ contains either $M_x^0$ or $M_x^1$.

\paragraph{\bf Properties of a dual certificate.}
Let $M$ admit a dual certificate  $\vec{\alpha}$, and consider some gadget~$G_\gamma$ in~$G$ of size~$s$. 
Recall that for each $(a,b) \in M$ complementary slackness for \ref{LP2} implies that $\alpha_a+\alpha_b=\wt_M(a,b)=0$.
Thus, either $\alpha_{a_\gamma^i}=-\alpha_{b_\gamma^i}$ for each $i \in [s]$, 
or $\alpha_{a_\gamma^i}=-\alpha_{b_\gamma^{i+1}}$ for each $i  \in [s]$.
Furthermore, let $\overline{M}_\gamma$ denote the matching $(M_\gamma^0 \cup M_\gamma^1) \setminus M$ 
(that is, $\overline{M}_\gamma=M_\gamma^1$ if $M_\gamma^0 \subseteq M$ and $\overline{M}_\gamma=M_\gamma^0$ otherwise).
Observe that half of the edges in $\overline{M}_\gamma$ have weight $+1$ and half of them have weight~$-1$ according to $\wt_M(\cdot)$, 
since exactly half of the agents in $G_\gamma$ prefer $M_\gamma^0$ to $M_\gamma^1$, while the other half prefer $M_\gamma^1$ to $M_\gamma^0$. 
Hence, using the feasibility of the dual certificate $\vec{\alpha}$, we get
\begin{align*}
0=& \sum_{a \in A_\gamma} \alpha_a + \alpha_{M(a)} = \sum_{a \in A_\gamma} \alpha_a + \sum_{b \in B_\gamma} \alpha_b \\
= &
\sum_{a \in A_\gamma} \alpha_a + \alpha_{\overline{M}_\gamma(a)}\geq |A_\gamma|/2 \cdot 1+|A_\gamma|/2 \cdot (-1)=0
\end{align*}
where the first equality holds because $\alpha_a+\alpha_b=0$ for each edge $(a,b) \in M$.
Thus, it follows that $\alpha_a+\alpha_b=\wt_M(a,b)$ must hold for each edge $(a,b) \in \overline{M}_\gamma$ as well. 
From this we get 
\begin{eqnarray}
\label{eqn-HA-dualvalues}
\notag
\alpha_{b_\gamma^{i+1}}&=&\alpha_{b_\gamma^{i+1}}+\alpha_{a_\gamma^i}-\alpha_{a_\gamma^i}-\alpha_{b_\gamma^i}+\alpha_{b_\gamma^i}=
\wt_{M}(a_\gamma^i,b_\gamma^{i+1})-\wt_{M}(a_\gamma^i,b_\gamma^i)+\alpha_{b_\gamma^i} \\ 
&=& \left\{ 
\begin{array}{ll}
\alpha_{b_\gamma^i} -1, & \textrm{\hspace{20pt}if $i \in \{0,1, \dots, 3s-1\}$,} \\
\alpha_{b_\gamma^i} +1, & \textrm{\hspace{20pt}if $i \in [6s-1] \setminus [3s-1]$,} \\
\end{array}
\right.
\end{eqnarray}
irrespectively of the type of gadget~$G_\gamma$. In particular, this implies
\begin{equation}
\label{eqn-HA-opposite-dualvalues}
\alpha_{b_\gamma^i}=\alpha_{b_\gamma^{-i}} 
\end{equation}
for each $i \in [3s]$.

Note that w.l.o.g. we can assume that $\alpha_{b_\ell^0}=0$, as otherwise we can decrease the value of~$\alpha_b$ for all objects $b$ by~$\alpha_{b_\ell^0}$ and increase $\alpha_a$ by the same amount for each agent~$a$. 
Therefore, (\ref{eqn-HA-dualvalues}) yields
\begin{equation}
\label{eqn-HA-levelgadget-values}
\alpha_{b_\ell^{-3}}=-3, \quad \alpha_{b_\ell^{-2}}=-2, \quad \alpha_{a_\ell^{-3}}=3
\end{equation}
where the last equality follows from $(a_\ell^{-3},b_\ell^{-3}) \in M$.

\paragraph{\bf Defining a vertex cover.}
We are going to define a set $S \subseteq V$, and show that $S$ is a vertex cover in~$H$ of size at most~$k$. 
Namely, let $S$ contain exactly those vertices~$x \in V$ for which \mbox{$M_x^0 \subseteq M$}. 
Note that since $G_x$ is a type-1 gadget of size~$n_v$ for each $x \in V$, 
the cost of all edges of~$M$ within~$G_x$ is either~$3n_v$ (if $x \in S$) or~$2n_v$ (if $x \notin S$).
Since the cost of $M_\ell^0 \subseteq M$ is $2n_\ell$, and 
by our choice of~$n_v$ we have
$\beta+n_v-3|E|=2n_\ell+(k+1)3n_v+(|V|-k-1)2n_v<\beta$,  
our budget implies $|S| \leq k$.
It remains to show that $S$ is indeed a vertex cover.

To this end, consider now the edges in $F_{\textup{bnd}}$.
Observe that $\wt_M(a_e^2,b_\ell^{-3})=1$, $\wt_M(a_e^4,b_\ell^{-2})=1$ and $\wt_M(a_\ell^{-3},b_e^4)=-1$ for any $e \in  E$.
Using (\ref{eqn-HA-levelgadget-values}) this implies
\begin{equation}
\label{eqn-HA-edgegadget-bounds}
\alpha_{a_e^2} \geq 4, \quad \alpha_{a_e^4} \geq 3, \quad \alpha_{b_e^4} \geq -4.
\end{equation}
By~(\ref{eqn-HA-opposite-dualvalues}) 
we also have  $\alpha_{b_e^2}=\alpha_{b_e^4}$.
Let us distinguish between two cases:
(i) first, if $M_e^0 \subseteq M$, then $\alpha_{a_e^2}+\alpha_{b_e^2} = \wt_M(a_e^2,b_e^2)=0$, 
which by (\ref{eqn-HA-edgegadget-bounds}) implies $\alpha_{b_e^2}\leq -4$;
(ii) second, if $M_e^1 \subseteq M$, then $\alpha_{a_e^4}+\alpha_{b_e^4} = \wt_M(a_e^4,b_e^4)=-1$, 
which by (\ref{eqn-HA-edgegadget-bounds}) implies $\alpha_{b_e^4}\leq -4$.
Hence, in either case we obtain $\alpha_{b_e^2}=\alpha_{b_e^4}=-4$ by~(\ref{eqn-HA-edgegadget-bounds}). 
This leads us to the following important facts: 
\begin{eqnarray}
\label{eqn-HA-edgegadget-value-M0}
\textrm{if $M_e^0 \subseteq M$, then }&& \!\! \ \alpha_{a_e^2}=4; \\
\label{eqn-HA-edgegadget-value-M1}
\textrm{if $M_e^1 \subseteq M$, then }&& \!\! \ \alpha_{a_e^3}=4. 
\end{eqnarray}

Now, consider the edges incident to some vertex gadget $G_x$. 
First, $\wt_M(a_x^{-3},b_\ell^{-2})=1$ implies $\alpha_{a_x^{-3}} \geq 3$ by the equations in (\ref{eqn-HA-levelgadget-values}).
From (\ref{eqn-HA-dualvalues}) it follows that
\begin{equation}
\label{eqn-HA-vertexgadget-value-M1}
\textrm{if $M_x^1 \subseteq M$, then } \alpha_{b_x^{-3}} = \alpha_{b_x^{-2}}-1 = \alpha_{M(a_x^{-3})}-1 = -\alpha_{a_x^{-3}}-1 \leq -4.
\end{equation}

We are now ready to show that $S$ is indeed a vertex cover in~$H$.
Let $e=(x,y) \in E$ be an edge in~$H$. 
First, if $M_e^0 \subseteq M$, 
then $\wt_M(a_e^2,b_x^{-3})=1$ together with~(\ref{eqn-HA-edgegadget-value-M0}) implies $\alpha_{b_x^{-3}} \geq -3$. 
Hence, by~(\ref{eqn-HA-vertexgadget-value-M1}) we must have $M_x^1 \not\subseteq M$, yielding $x \in S$. 
Second, if $M_e^1 \subseteq M$, 
then $\wt_M(a_e^3,b_y^{-3})=1$ together with~(\ref{eqn-HA-edgegadget-value-M1}) implies $\alpha_{b_y^{-3}} \geq -3$. 
By (\ref{eqn-HA-vertexgadget-value-M1}) applied to gadget~$G_y$ we obtain $M_y^1 \not\subseteq M$, yielding $y \in S$. 
Thus, we can conclude that either $x \in S$ or $y \in S$, proving our claim. 

\paragraph{\bf Direction ``$\Leftarrow$''.}
For the other direction, given a vertex cover $S \subseteq V$ of size at most $k$ in $H$, we will show that 
a popular assignment $M$ of total cost at most~$\beta$ exists in $G$. 
For each edge $e \in E$, let us fix one of its endpoints in $S$, 
and denote it by $\tau(e)$. 
We may define $M$ as follows:

\[
M=M_\ell^0 \cup \left( \bigcup_{x \in S} M_x^0 \right) \cup \left(\bigcup_{x \in V \setminus S} M_x^1 \right)
	\cup \left( \bigcup_{\substack{e=(x,y) \in E, \\ \tau(e)=x}} M_e^0 \right)
	\cup \left( \bigcup_{\substack{e=(x,y) \in E, \\ \tau(e)=y}} M_e^1 \right).
\]

\begin{figure}[t]
\begin{center}	
\includegraphics[width=\linewidth]{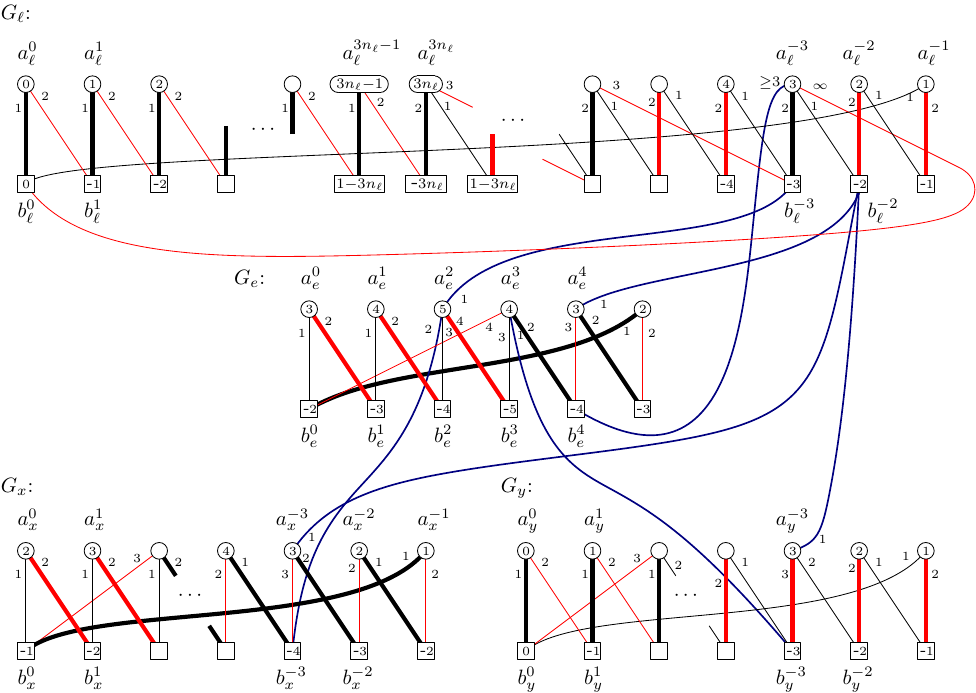}
\caption{An illustration of a popular assignment for the instance $G$ constructed in the proof of Theorem~\ref{thm:housealloc-nphard}, 
given a vertex cover $S$.
The figure assumes $e=(x,y) \in E$ with $x \notin S$ and $y \in S$. 
Edges contained in $M$ are depicted in bold. The values of a dual certificate $\vec{\alpha}$ for $M$ 
are indicated within the corresponding agent or object.
}
\label{fig:HA-allocation}
\end{center}
\end{figure}

It is clear that $M$ indeed has total cost at most~$\beta$: the cost of $M_\ell^0$ is $2n_\ell$, 
the cost of $\bigcup_{x \in S} M_x^0$ is~$3kn_v$,  
the cost of $\bigcup_{x \in V\setminus S} M_x^1$ is $2(|V|-k)n_v$, 
while the cost of all edges of $M$ within edge gadgets is at most $3|E|$.
To show that $M$ is popular, we define a dual certificate~$\vec{\alpha}$ for~$M$ by defining $\alpha_b$ for each object $b \in B$ as follows; 
we set $\alpha_a = - \alpha_{M(a)}$ for each agent~$a \in A$.
See Fig.~\ref{fig:HA-allocation} for an illustration.

\medskip
\begin{tabular}{ll}
$\alpha_{b_\ell^i}=-|i|$ & for each $i \in \{-3n_\ell+1, \dots, -1, 0,1,\dots, 3n_\ell\}$; \\
$\alpha_{b_x^i}=-|i|$ & for each $x \in S$ and $i \in \{-3n_v+1, \dots, -1, 0,1,\dots, 3n_v\}$; \\
$\alpha_{b_x^i}=-|i|-1$ & for each $x \in V \setminus S$ and $i \in \{-3n_v+1, \dots, -1, 0,1,\dots, 3n_v\}$; \\
$\alpha_{b_e^i}=-|i|-2$ & for each $e \in E$ and $i \in \{-2,-1,0,1,2, 3\}$; 
\end{tabular}
\medskip

This finishes the proof of the theorem. \qed
\end{proof}

\section{Open problems}\label{sec:open-problems}
We proposed a polynomial-time algorithm for computing a popular assignment in an instance $G = (A\cup B,E)$ with one-sided preferences,
if one exists.
Our algorithm solves $O(n^2)$ instances of the maximum matching problem in certain subgraphs of~$G$ (where $|A| = |B| = n$) and
it is easy to show instances where our algorithm indeed solves $\Theta(n^2)$ instances of the maximum matching problem.
The running time of our algorithm is $O(m\cdot n^{5/2})$ (where $|E| = m$) and 
by using the almost-linear time algorithm for maximum flow~\cite{CKLPPS22} to solve the bipartite maximum matching problem, we can 
get a faster randomized algorithm for the \myproblem{popular assignment} problem. Can we 
show faster deterministic or randomized algorithms for the \myproblem{popular assignment} problem 
by solving $o(n^2)$ instances of the maximum matching problem or via some other approach?

Another open problem is the computational complexity of the \myproblem{minimum-cost popular matching} problem when agents' preferences are partial orders.
For weak rankings, a polynomial-time algorithm for this problem follows from the combinatorial characterization of popular matchings given in \cite{AIKM07}. However as illustrated in an example in Section~\ref{app:characterization-of-pop-matchings}, this characterization does not hold for partial order preferences. 
So the computational complexity of the \myproblem{minimum-cost popular matching} problem for partial order preferences is open.

One more open problem is to show a short {\em witness} that a given instance $G$ does not admit a popular assignment. Rather than run our
algorithm and discover that $G$ has no popular assignment, is there a forbidden structure that causes $G$ not to admit a popular assignment?
Can we characterize instances that admit popular assignments? Interestingly, such a result is known for the stable roommates problem~\cite{Tan91}; 
recall our discussion in Section~\ref{sec:intro} on the similarity between results for the \myproblem{popular assignment} problem and the 
\myproblem{stable roommates} problem.

Our $\mathsf{W}_l[1]$-hardness proof for the \myproblem{$k$-unpopularity margin}  problem with parameter~$k$ needs weak rankings. We are able to show that this problem remains $\mathsf{NP}$-hard for strict rankings (Lemma~\ref{lem:reduction-weak-to-strict}). 
However, the parameterized complexity of this case is still open: is the \myproblem{$k$-unpopularity margin} problem in FPT with parameter~$k$, if agents' preferences are strict rankings?

\section*{Acknowledgments}
Telikepalli Kavitha is supported by the DAE, Government of India, under project no. RTI4001. 
Tam\'{a}s Kir\'{a}ly is supported by the Ministry of Innovation and Technology of Hungary from the National Research, Development and Innovation Fund, financed under the ELTE TKP 2021-NKTA-62 funding scheme. 
Jannik Matuschke is supported by internal funds of KU Leuven. 
Ildik\'o Schlotter is supported by the Hungarian Academy of Sciences under its Momentum Programme (\hbox{LP2021-2}) and its J\'anos Bolyai Research Scholarship, 
and by the Hungarian Scientific Research Fund (OTKA grant  K124171). Ulrike Schmidt-Kraepelin was supported by the Deutsche For\-schungs\-gemeinschaft (DFG) under grant BR~4744/2-1 while affiliated with TU Berlin, and the \emph{National Science Foundation (NSF)} under Grant No. DMS-1928930 and by the Alfred P. Sloan Foundation under grant G-2021-16778, while she was in residence at the Simons Laufer Mathematical Sciences Institute (formerly MSRI) in Berkeley, California, during the Fall 2023 semester.

\bibliographystyle{abbrv}
\bibliography{popassign}

\section*{Appendix} The proofs of Proposition~\ref{prop:penalty-truncation} and Lemma~\ref{lem:reduction-weak-to-strict} are given below;
for convenience, we restate their statements.

\proppenalty*
\begin{proof}
     Define $\k'=\max \{ |\alpha_b|:b \in B'\}$ which also equals $
     \max \{ \alpha_a:a \in A'\}$ by complementary slackness. 
     Suppose for contradiction that $\k'>\k$, and let $B_{\max}$ denote the set of all objects~$b$ with~$\alpha_b=-\k'$. We choose $\vec{\alpha}$ so that $\k'$ is minimized, and subject to that, $|B_{\max}|$ is as small as possible. Let us call an edge~$(a,b)$ in~$G$ \emph{tight}, if $\alpha_a+\alpha_b=\wt_M(a,b)$; clearly, all edges of~$M$ are tight by complementary slackness. For each object~$b \in B_{\max}$, let us define the subgraph $H_b$ of $G'$ spanned by all $M'$-alternating paths in $G'$ starting at~$b$ with an $M'$-edge (if any) that only contain tight edges. 
     By the same arguments as in the proof of Theorem~\ref{thm:pop-w-penalty}, we know that $H_b$ must contain an agent whose $\alpha$-value is zero, otherwise we could increase the $\alpha$-value of each object in $H_b$ by~$1$ and decrease the $\alpha$-value of each agent in~$H_b$ by~1, contradicting our choice of~$\vec{\alpha}$.

    First, suppose that $l_{i}(a) \in B_{\max}$ 
    for some $a \in A$ and $i \in [\k]$. We claim $i=\k$.
    On the one hand, if $i<\k$ and $(p_i(a),l_{i+1}(a)) \in M'$, then by $\wt_{M'}(p_i(a),l_i(a))=1$ we know $\alpha_{p_i(a)} \geq \k'+1$, a contradiction to the definition of~$\k'$.
    On the other hand, if $i<\k$ and  $(p_i(a),l_i(a)) \in M'$, let $a'$ be an agent in~$H_{l_i(a)}$ with $\alpha_{a'} = 0$.
    Note that $a' \neq p_i(a)$ because $\alpha_{p_i(a)} = - \alpha_{l_i(a)} = \k' > 0$.
    Hence the $M'$-alternating path certifying that~$a'$ is in~$H_{l_i(a)}$ must begin with the edges $(l_i(a), p_i(a))$ and $(p_i(a), l_{i+1}(a))$.
    In particular,
    these two edges must be tight, yielding $\alpha_{p_{i}(a)}=\k'$ and $\alpha_{l_{i+1}(a)}=-\k'-1$ due to $\wt_{M'}(p_i(a),l_{i+1}(a))=-1$, again a contradiction. Thus $l_{i}(a) \in B_{\max}$ implies $i = \k$.
    
    Consequently, we know $B_{\max} \subseteq L_\k \cup B$. Thus, for each dummy $d \in D$ we must have $\alpha_d \geq \k'$, since $d$ is adjacent to all objects in~$L_\k \cup B$ and indifferent between them. By the definition of~$\k'$ we get $\alpha_d = \k'$. 

    Consider any agent~$a \in A$ for which $M'(a) \in B$. Then we have $M'(l_\k(a)) \in D$, which in turn yields $\alpha_{l_\k(a)}=-\k'$. Since each edge in $P(a) \setminus M'$ has $\wt_{M'}$-weight $-1$, for each $i=1,\dots, \k-1$ we iteratively obtain that $\alpha_{p_{\k-i}(a)} \geq \k'-i$ and $\alpha_{l_{\k-i}(a)} \leq -\k'+i$, yielding also 
    $\alpha_a \geq \k'-\k \geq 1$.
    Hence, all agents on~$P(a)$ have positive $\alpha$-values.

    Now let $b \in B_{\max}$ and let $a_0$  be an agent in~$H_{b}$ with $\alpha_{a_0}=0$. We know that $a_0$ is not a dummy (since all dummies have $\alpha$-value~$\k'$), and by the previous paragraph, $a_0$ does not lie on any path~$P(a')$ with $M'(a') \in B$. Hence, $a'$ must lie on a path~$P(a')$ for which $M'(a')=l_1(a)$.
    Let $Q$ be an $M'$-alternating path in~ $H_{b}$ from $b$ to~$a_0$. Since $Q$ starts with an edge of~$M'$ and is $M'$-alternating,  the fact $M'(a') \notin B$ implies that $Q$ must enter~$P(a')$ from some dummy~$d$, proceeding through the subpath of $P(a')$ from $l_\k(a')$ to~$a_0$. However, by $\alpha_d=-\k'$ and the tightness of the edges of~$Q$ we obtain that all agents on~$P(a')$ have positive $\alpha$-value, a contradiction to the definition of~$a_0$. This proves the proposition. \qed
\end{proof}

\lemweaktostrict*
\begin{proof}
  Let $G$ be an instance with weak rankings. We divide the proof into two parts. In the first part we show that we can assume without loss of generality that the set of agents in $G$ can be partitioned into two groups, $A_{\succ}$ and $A_{\sim}$, where agents in $A_{\succ}$ have strict preferences over the objects in their neighborhood and agents in $A_{\sim}$ are indifferent among all their neighboring objects. Starting from such an instance, we then prove the lemma.   

\paragraph{\bf Part I.} Starting from a graph $G$ with weak rankings, we create a graph $\hat{G}$ and preferences as follows: For every agent $v$ in $G$, let $B_v^{(1)} \succ \dots \succ B_v^{(r_v)}, r_v \in \{1,\dots,m\}$ be the weak ranking over its neighboring objects. We add agent~$v$ to~$\hat{G}$, and for each $i \in \{1,\dots,r_v\}$ we further add a new agent~$a_{v}^{(i)}$ and a new object~$b_{v}^{(i)}$ to $\hat{G}$. The new agent~$a_{v}^{(i)}$ is connected to---and defined to be indifferent among---all objects in $B_v^{(i)} \cup \{b_{v}^{(i)}\}$. Lastly, we introduce edges from~$v$ to all objects in~$\bigcup_{i=1}^{r_v} b_v^{(i)}$, with preferences $b_{v}^{(1)} \succ b_{v}^{(2)} \succ \dots \succ b_{v}^{(r_v)}$. In $\hat{G}$ we can partition the set of nodes into~$A_{\succ}$, containing copies of agents in $G$, and $A_{\sim}$, containing the newly introduced agents. Agents in~$A_{\succ}$ have strict preferences and all agents in $A_{\sim}$ are indifferent among all their neighbors. Below, we show that the two instances are essentially equivalent.

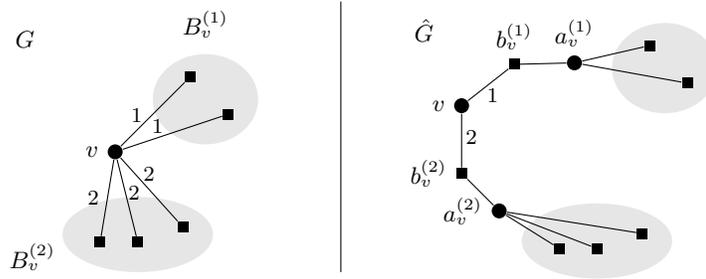
\begin{figure}
    \centering
\begin{tikzpicture}
    \draw[fill=black!10,draw=none] (0.3,-1.1) ellipse (1cm and .5cm);
    \draw[fill=black!10,draw=none] (1.2,.7) ellipse (.7cm and .6cm);
    \node[circle,fill=black,label=left:$v$,inner sep = 2pt] (v) at (0,0){};
    \node[rectangle, fill=black,inner sep = 2pt] (b1) at (1,1){};
    \node[rectangle, fill=black,inner sep = 2pt] (b2) at (1.5,.5){};
    \node[rectangle, fill=black,inner sep = 2pt] (b3) at (-.2,-1.2){};
    \node[rectangle, fill=black,inner sep = 2pt] (b4) at (.3,-1.2){};
    \node[rectangle, fill=black,inner sep = 2pt] (b5) at (.9,-1){};
    \draw (v) -- (b1) node [near start, above] {\scriptsize $1$};
    \draw (v) -- (b2) node [shift={(190:0.95)}] {\scriptsize $1$};
    \draw (v) -- (b3) node [midway, left] {\scriptsize $2$};
    \draw (v) -- (b4) node [shift={(110:0.7)}, right] {\scriptsize $2$};
    \draw (v) -- (b5) node [near start, right] {\scriptsize $2$};
    \node at (-1.1,-1.4) {$B_v^{(2)}$};
    \node at (1.2,1.7) {$B_v^{(1)}$};
    \draw (3,-1.6) -- (3,2); 
    \node at (-1.2,1.5) {$G$};
    \end{tikzpicture}
    \begin{tikzpicture}\hspace{.7cm}
    \draw[fill=black!10,draw=none] (1.8,-1.6) ellipse (1cm and .5cm);
    \draw[fill=black!10,draw=none] (2.7,.7) ellipse (.7cm and .6cm);
    \node[circle,fill=black,label=left:$v$,inner sep = 2pt] (v) at (0,0.2){};
    \node[rectangle, fill=black,inner sep = 2pt,label=above:$b_v^{(1)}$] (c1) at (.7,.75){};
    \node[circle, fill=black,inner sep = 2pt,label=above:$a_v^{(1)}$] (d1) at (1.5,.77){};
    \node[rectangle,label=left:$b_v^{(2)}$, fill=black,inner sep = 2pt] (c2) at (0,-.7){};
    \node[circle, fill=black,label=left:$a_v^{(2)}$, inner sep = 2pt] (d2) at (0.5,-1.2){};
    \node[rectangle, fill=black,inner sep = 2pt] (b1) at (2.5,1){};
    \node[rectangle, fill=black,inner sep = 2pt] (b2) at (3,.5){};
    \node[rectangle, fill=black,inner sep = 2pt] (b3) at (1.3,-1.7){};
    \node[rectangle, fill=black,inner sep = 2pt] (b4) at (1.8,-1.7){};
    \node[rectangle, fill=black,inner sep = 2pt] (b5) at (2.4,-1.5){};
    \draw (v) -- (c1) node [shift={(235:0.5)}] {\scriptsize $1$};
    \draw (v) -- (c2) node [shift={(75:0.5)}] {\scriptsize $2$};
    \draw (c1) -- (d1);
    \draw (c2) -- (d2);
    \draw (d1) -- (b1);
    \draw (d1) -- (b2);
    \draw (d2) -- (b3);
    \draw (d2) -- (b4);
    \draw (d2) -- (b5);
    \node at (-.5,1.2) {$\hat{G}$};
    \end{tikzpicture}
    \caption{Illustration of the first part of the proof of Lemma~\ref{lem:reduction-weak-to-strict}. The left side illustrates the neighborhood of a fixed node $v$ within the original graph $G$. The right side illustrates the corresponding situation within the graph~$\hat{G}$. Agents are depicted by circles and objects by squares. Labels of the edges indicate the rank of the edge within the ranking of the incident agent. Agents $a_{v}^{(1)}$ and $a_{v}^{(2)}$ are indifferent among all their neighbors, hence, their labels are omitted. }
    \label{fig:reduction-weak-to-structuredWeak}
\end{figure}

More precisely, we show that there exists a bijection $f$ mapping assignments in $G$ to assignments in $\hat{G}$ such that $\Delta(N,M) = \Delta(f(N),f(M))$ for any two assignments $M$ and $N$ in $G$. Let $M$ be an assignment in $G$.
We start with $f(M) = \emptyset$. For every edge $(v,w) \in M$ we do the following: Let the preferences of~$v$ be~$B_v^{(1)} \succ \dots \succ B_v^{(r_v)}$, and let $i \in \{1,\dots,r_v\}$ be such that $w \in B_v^{(i)}$. Then, we add the edges $(v,b_v^{(i)})$ and $(a_v^{(i)}, w )$ to $f(M)$ (see Fig.~\ref{fig:reduction-weak-to-structuredWeak}). Moreover, for all indices $j \in \{1, \dots, r_v\} \setminus \{i\}$ we add the edge $(a_v^{(j)},b_v^{(j)})$ to $f(M)$. It is easy to see that this is a bijection. 

Since for every agent in $A_{\succ}$ the rank of its partner in~$M$ equals the rank of its partner in~$f(M)$ and agents in~$A_{\sim}$ are indifferent among all their neighbors, we get $\Delta(N,M) = \Delta(f(N),f(M))$.  

\paragraph{\bf Part II.} Due to part I, we can assume that the agents in $G$ are partitioned into sets~$A_{\succ}$ and~$A_{\sim}$ such that agents in $A_{\succ}$ have strict preferences over objects and the agents in $A_{\sim}$ are indifferent among all their neighboring objects. Starting from the graph~$G$, we create a graph~$G'$ with strict rankings as follows. We first copy all agents and objects to $G'$. Then for each $a \in A_{\sim}$ we introduce two new agents $a'$ and $a''$ and two new objects $b'_a$ and~$b''_a$. We add all possible edges from $\{a,a',a''\}$ to $\Nbr_{G}(a) \cup \{b'_{a},b''_{a}\}$ (see Fig.~\ref{fig:my_label}). The preferences of agents $a,a'$, and $a''$ are identical, namely, $b'_{a}$ is their first choice, $b''_{a}$ their second choice, followed by all objects in $\Nbr_G(a)$ in arbitrary order. 

We define a function $f$ which maps assignments in $G'$ to assignments in $G$. Let $M'$ be an assignment in $G'$. For each $(a,b) \in M'$ where $a \in A_{\succ}$, we add the edge $(a,b)$ to $f(M')$. Now, observe that for every $a \in A_{\sim}$, exactly one of the nodes $a,a',a''$ is matched to a node of the original neighborhood, i.e., $\Nbr_{G}(a)$. Let $b$ be the assigned object. Then we add $(a,b)$ to $f(M')$. It is easy to see that $f(M')$ is then an assignment in $G$. 
We continue by observing that 
\begin{enumerate}[leftmargin=24pt]
    \myitem[(i)]\label{item:fsurjective} $f$ is surjective, and
    \myitem[(ii)]\label{item:equalvotes} for all $ a \in A_{\succ}$ and two assignment $M'$, $N'$ in $G'$: $M' \succ_a N'$ if and only if $f(M') \succ_a f(N')$.     
\end{enumerate}

For \ref{item:fsurjective}, note that for every assignment $M$ in $G$ we can create an assignment $M'$ in $G'$ such that $f(M')=M$ holds, as follows: Copy all edges from $M$ to $M'$ and then add for every $a \in A_{\sim}$ the edges $(a',b_a')$ and $(a'',b_a'')$. Clearly, $f(M') = M$. For \ref{item:equalvotes} observe that for all assignments $M'$ in $G'$, agents in $A_{\succ}$ have the same assigned object in $M'$ as in $f(M')$. 

We define $q= |A_{\sim}|$ and show that $G$ admits an assignment with unpopularity margin at most~$k$ if and only if $G'$ admits an assignment with unpopularity margin at most~$k+q$,
for any $k \in \mathbb{N}$.  

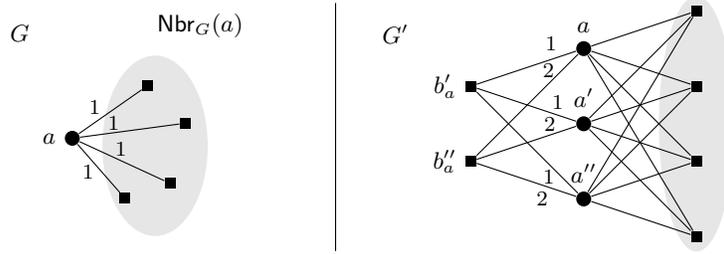
\begin{figure}
    \centering
    \begin{tikzpicture}
    \draw[fill=black!10,draw=none] (1.1,.2) ellipse (.7cm and 1.2cm);
    \node at (1.7,1.8) {$\Nbr_G(a)$}; 
    \node[circle,fill=black,label=left:$a$,inner sep = 2pt] (a) at (0,.3){};
    \node[rectangle, fill=black,inner sep = 2pt] (b1) at (1,1){};
    \node[rectangle, fill=black,inner sep = 2pt] (b2) at (1.5,.5){};
    \node[rectangle, fill=black,inner sep = 2pt] (b3) at (1.3,-.3){};
    \node[rectangle, fill=black,inner sep = 2pt] (b4) at (.7,-.5){};
     \draw (a) -- (b1) node [near start, above] {\scriptsize $1$};
    \draw (a) -- (b2) node [shift={(180:0.95)}] {\scriptsize $1$};
    \draw (a) -- (b3) node [shift={(145:0.8)}] {\scriptsize $1$};
    \draw (a) -- (b4) node [near start, below] {\scriptsize $1$};
     \draw (3.5,-1.2) -- (3.5,2.1); 
     \node at (-.7,1.7){$G$};
    \end{tikzpicture}
    \begin{tikzpicture}\hspace*{.4cm}
    \node at (-2.5,1.2){$G'$};
    \draw[fill=black!10,draw=none] (1.5,0) ellipse (.5cm and 1.7cm);
    \node[circle,fill=black,label=above:$a$,inner sep = 2pt] (a) at (0,1){};
    \node[circle, fill=black,label=above:$a'$,inner sep = 2pt] (ap) at (0,0){};
    \node[circle, fill=black,label=above:$a''$,inner sep = 2pt] (app) at (0,-1){};
    \node[rectangle, fill=black,label=left:$b'_a$,inner sep = 2pt] (bp) at (-1.5,.5){};
    \node[rectangle, fill=black,label=left:$b''_a$,inner sep = 2pt] (bpp) at (-1.5,-.5){};
    \node[rectangle, fill=black,inner sep = 2pt] (b1) at (1.5,1.5){};
    \node[rectangle, fill=black,inner sep = 2pt] (b2) at (1.5,.5){};
    \node[rectangle, fill=black,inner sep = 2pt] (b3) at (1.5,-.5){};
    \node[rectangle, fill=black,inner sep = 2pt] (b4) at (1.5,-1.5){};
    \draw (a) -- (bp) node [near start, above] {\scriptsize $1$};
    \draw (bp) -- (ap) node [shift={(140:0.45)}] {\scriptsize $1$};
    \draw (bp) -- (app) node [shift={(147:0.55)}] {\scriptsize $1$};
    \draw (bpp) -- (a) node [shift={(212:0.55)}] {\scriptsize $2$};
    \draw (bpp) -- (ap) node [shift={(180:0.45)}] {\scriptsize $2$};
    \draw (bpp) -- (app) node [shift={(180:0.55)}] {\scriptsize $2$};
    \draw (a) -- (b1);
    \draw (a) -- (b2);
    \draw (a) -- (b3);
    \draw (a) -- (b4);
    \draw (ap) -- (b1);
    \draw (ap) -- (b2);
    \draw (ap) -- (b3);
    \draw (ap) -- (b4);
    \draw (app) -- (b1);
    \draw (app) -- (b2);
    \draw (app) -- (b3);
    \draw (app) -- (b4);
    \end{tikzpicture}
            \caption{Illustration of the second part of the proof of Lemma~\ref{lem:reduction-weak-to-strict}. The left side illustrates the neighborhood of an agent~$a$, indifferent among all its neighbors, within the graph~$G$. The right side captures the corresponding gadget in the graph~$G'$. Labels on the edges indicate the preferences of agents. The ranks of edges between $\{a,a',a''\}$ and $\Nbr_G(a)$ can be chosen arbitrarily (but need to be larger than $2$), hence, these labels are omitted.}
    \label{fig:my_label}
\end{figure}

\paragraph{\bf Direction ``$\Rightarrow$''.} Assume that $G$ admits an assignment~$M$ with unpopularity margin at most~$k$. Choose an assignment~$M'$ from~$G'$ such that $f(M')=M$ holds (such an $M'$ is guaranteed to exist by \ref{item:fsurjective}). Let $N'$ be an assignment in~$G'$ maximizing $\Delta(N',M')$, so $\mu(M')=\Delta(N',M')$. Since $\Delta(f(N'),f(M')) \leq \max_N \Delta(N,M) = k$, we know by \ref{item:equalvotes} that agents in $A_{\succ}$ contribute at most~$k$ to $\Delta(N',M')$. Moreover, we claim that agents not in $A_{\succ}$ contribute at most $q$ to $\Delta(N',M')$. To see this, consider the gadget for some agent~$a \in A_{\sim}$. We distinguish two cases. First, assume that the agent from $\{a,a',a''\}$ which is assigned an object in $\Nbr_G(a)$ is the same in assignments $M'$ and~$N'$; w.l.o.g. we assume it is agent~$a$. Then $a'$ and $a''$ together contribute $0$, and $a$ at most~$1$ to~$\Delta(N',M')$. Second, assume w.l.o.g. that $a$ is assigned some object in $\Nbr_G(a)$ by~$M'$, while $a'$ is assigned some object in~$\Nbr_G(a)$ by~$N'$. Then $a$ and $a'$ together contribute~$0$,and $a''$ at most~$1$. As this holds for every gadget belonging to agents in $A_{\sim}$, this proves $\Delta(N',M') \leq k+q$. 

\paragraph{\bf Direction ``$\Leftarrow$''.}
For the other direction, assume that $G'$ admits an assignment $M'$ with unpopularity margin at most $k+q$. We claim that $f(M')$ has unpopularity at most $k$. Assume for contradiction that there exists some assignment $N$ in $G$ with $\Delta(N,f(M')) > k$. We construct $N'$ as follows. For every agent~$a \in A_{\succ}$, we let $N'(a)=N(a)$. Next, for every agent~$a \in A_{\sim}$, let $b_{M}$ be its assigned object in $f(M')$ and $b_N$ be its assigned object in $N$. We can assume w.l.o.g. that $a$ is matched to~$b_M$ in~$M'$, $a'$ is matched to~$b'_a$, and $a''$ is matched to~$b''_a$. In~$N'$ we match $a'$ to~$b_N$, $a''$ to~$b'_a$, and $a$ to~$b''_a$. Then $a$, $a'$, and $a''$ together contribute exactly~$1$ to~$\Delta(N',M')$. Using the same argument for all $a \in A_{\sim}$ and \ref{item:equalvotes} yields that $\mu(M') \geq \Delta(N',M') > k + q$, a contradiction. \qed
\end{proof}
\end{document}